\documentclass[letterpaper,11pt]{article}
\usepackage{microtype}
\usepackage[letterpaper,margin=1in]{geometry}
\usepackage[numbers,sort&compress]{natbib}
\usepackage[table]{xcolor}
\usepackage{amsmath}
\usepackage{amssymb}
\usepackage{mathrsfs}
\usepackage{amsthm}
\usepackage{enumitem}
\usepackage{textgreek}
\usepackage{graphicx}
\graphicspath{{./}{./figures/}{./figs/}{../}{../figures/}{../figs/}{../../}{../../figures/}{../../figs/}}
\usepackage{framed}
\usepackage[small]{caption}
\usepackage{booktabs}
\usepackage{tikz-cd}
\usepackage{hyperref}
\usepackage{doi}
\usepackage[framemethod=tikz]{mdframed}
\usepackage{thm-restate}
\usepackage{mathtools}
\usepackage{xspace}
\usepackage[capitalize, nameinlink]{cleveref}
\usepackage{multirow}
\usepackage{cellspace}
\usepackage{array}
\usepackage{tcolorbox}
\usepackage{orcidlink}

\usepackage{algorithm}
\usepackage[noend]{algpseudocode}
\makeatletter
\algnewcommand{\LineComment}[1]{\Statex \hskip\ALG@thistlm \(\triangleright\) #1}
\algnewcommand{\Statey}{\Statex \hskip\ALG@thistlm}
\makeatother

\Crefname{remark}{Remark}{Remarks}
\Crefname{observation}{Observation}{Observations}

\theoremstyle{plain}
\newtheorem{theorem}{Theorem}[section]
\newtheorem{lemma}[theorem]{Lemma}
\newtheorem{proposition}[theorem]{Proposition}
\newtheorem{corollary}[theorem]{Corollary}

\theoremstyle{definition}
\newtheorem{definition}[theorem]{Definition}

\theoremstyle{plain}

\theoremstyle{remark}

\newcommand{\fB}{\mathcal{B}}

\newcommand{\fF}{\mathcal{F}}

\newcommand{\fR}{\mathcal{R}}

\newcommand{\imax}{i_{\max}}
\newcommand{\dlog}{d_{\mathsf{log}}}
\newcommand{\dlogsup}[1]{\dlog^{(#1)}}

\newcommand{\Otilde}{\widetilde{O}}

\DeclareMathOperator{\poly}{poly}
\DeclareMathOperator{\dist}{dist}

\DeclareMathOperator{\col}{col}
\DeclareMathOperator{\indeg}{indeg}
\DeclareMathOperator{\outdeg}{outdeg}

\newcommand{\algoname}[1]{\textnormal{{\textsc{#1}}}}\newcommand{\EliminationWalk}{\algoname{EliminationWalk}}
\newcommand{\RulingSubgraphs}{\algoname{RulingSubgraphs}}
\newcommand{\RulingOrchids}{\algoname{RulingOrchids}}
\newcommand{\DirectionalEliminationWalk}{\algoname{DirectionalEliminationWalk}}
\newcommand{\PausingEliminationWalk}{\algoname{PausingEliminationWalk}}

\newcommand{\OutdegreeReduction}{\algoname{OutdegreeReduction}}

\newcommand{\LDelEC}{L$\Delta$EC}
\newcommand{\NLDelEC}{NL$\Delta$EC}

   \newcommand{\Tlayer}{T_{\mathsf{layercol}}}  \newcommand{\Talg}[1]{T_{\hyperref[#1]{\textsf{Alg.\ref*{#1}}}}}
\newcommand{\Trsg}{\Talg{alg:rulingsubgraphs}}

\newcommand{\LOCAL}{LOCAL\xspace}
\newcommand{\CONGEST}{CONGEST\xspace}

\DeclarePairedDelimiter{\set}{\{}{\}}
\DeclarePairedDelimiter{\card}{\lvert}{\rvert}
\DeclarePairedDelimiter{\floor}{\lfloor}{\rfloor}
\DeclarePairedDelimiter{\ceil}{\lceil}{\rceil}
\DeclarePairedDelimiter{\parens}{(}{)}

\newenvironment{mycover}
{\list{}{\listparindent 0pt
        \itemindent    \listparindent
        \leftmargin    1cm
        \rightmargin   1cm
        \parsep        0pt}\raggedright
    \item\relax}
{\endlist}

\newcommand{\myorcid}[1]{{\small \orcidlink{#1}}}
\newcommand{\myaff}[1]{\,$\cdot$\, {\small #1}\par\smallskip}

\definecolor{darkgreen}{rgb}{0,0.5,0}
\definecolor{darkred}{rgb}{0.4,0,0}
\hypersetup{
    colorlinks=true,
    linkcolor=darkred,
    citecolor=darkgreen,
    filecolor=black,
    urlcolor=[rgb]{0,0.1,0.5},
    pdftitle={Faster Distributed Δ-Coloring via Ruling Subgraphs},
    pdfauthor={Yann Bourreau, Sebastian Brandt, Alexandre Nolin}
}

\begin{document}

\setcounter{page}{0}
\thispagestyle{empty}

\begin{mycover}
    {\huge\bfseries\boldmath Faster Distributed $\Delta$-Coloring via Ruling Subgraphs\par}
    \bigskip
    \bigskip
    \bigskip
    \textbf{Yann Bourreau}
    \myorcid{0009-0001-1819-8348}
    \myaff{CISPA Helmholtz Center for Information Security}

    \textbf{Sebastian Brandt}
    \myorcid{0000-0001-5393-6636}
    \myaff{CISPA Helmholtz Center for Information Security}
    
    \textbf{Alexandre Nolin}
    \myorcid{0000-0002-3952-0586}
    \myaff{CISPA Helmholtz Center for Information Security}
\end{mycover}
\bigskip

\begin{abstract}
	Brooks' theorem states that all connected graphs but odd cycles and cliques can be colored with $\Delta$ colors, where $\Delta$ is the maximum degree of the graph. Such colorings have been shown to admit non-trivial distributed algorithms [Panconesi and Srinivasan, Combinatorica 1995]\nocite{panconesi95delta} and have been studied intensively in the distributed literature.
    In particular, it is known that any deterministic algorithm computing a $\Delta$-coloring requires $\Omega(\log n)$ rounds in the LOCAL model [Chang, Kopelowitz, and Pettie, FOCS 2016], and that this lower bound holds already on constant-degree graphs.
    In contrast, the best upper bound in this setting is given by an $O(\log^2 n)$-round deterministic algorithm that can be inferred already from the works of [Awerbuch, Goldberg, Luby, and Plotkin, FOCS 1989] and [Panconesi and Srinivasan, Combinatorica 1995] roughly three decades ago, raising the fundamental question about the true complexity of $\Delta$-coloring in the constant-degree setting.

    We answer this long-standing question almost completely by providing an almost-optimal deterministic $O(\log n \log^* n)$-round algorithm for $\Delta$-coloring, matching the lower bound up to a $\log^* n$-factor.
    Similarly, in the randomized LOCAL model, we provide an $O(\log \log n \log^* n)$-round algorithm, improving over the state-of-the-art upper bound of $O(\log^2 \log n)$ [Ghaffari, Hirvonen, Kuhn, and Maus, Distributed Computing 2021]\nocite{GHKM_dc21} and almost matching the $\Omega(\log \log n)$-round lower bound by [BFHKLRSU, STOC 2016].

    Our results make progress on several important open problems and conjectures.
    One key ingredient for obtaining our results is the introduction of \emph{ruling subgraph families} as a novel tool for breaking symmetry between substructures of a graph, which we expect to be of independent interest.
\end{abstract}

\clearpage
\setcounter{page}{0}
\thispagestyle{empty}

\tableofcontents
\clearpage

\section{Introduction}
Since the dawn of the theoretical study of distributed algorithms in the 80s \cite{linial1987distributive}, two central goals of this area have been to improve the known complexity bounds for fundamental problems and develop generic techniques for designing faster and faster algorithms.
In this work, we will contribute to both of these goals by developing a new and faster technique for an essential objective in distributed computation---breaking symmetry between large substructures of the input instance---and use this technique to design distributed algorithms that obtain improved complexity bounds for the fundamental problem of $\Delta$-coloring (where $\Delta$ denotes the maximum degree of the input graph).

\paragraph{Model of computation.}
We consider the standard LOCAL model of distributed computation~\cite{linial1992locality,peleg2000}. In the LOCAL model, a network is modeled by a graph, whose nodes are computational entities that communicate with each other to solve a given graph problem.
In the beginning of the computation, each node is only aware of its own degree and a globally unique identifier assigned to it.
The distributed computation then proceeds in synchronous rounds: in each round each node can send an arbitrarily large message to each of its neighbors, and, after receiving the messages from its neighbors, perform some internal computation.
Each node has to decide to terminate at some point, upon which it ceases participation in the communication and outputs its local part of the computed solution; if, for instance, the task is to compute a proper coloring, then each node has to output a color such that the resulting coloring is indeed proper.
The complexity of such a distributed algorithm is simply the number of rounds until the last node terminates; the complexity of a problem is the complexity of an optimal algorithm for the respective problem.
We will use $n$ and $\Delta$ to refer to the number of nodes and the maximum degree of the input graph, respectively.
The related CONGEST model~\cite{peleg2000}, which we will mention on a few rare occasions, additionally imposes a maximum message size of $O(\log n)$ bits.

\paragraph{A fundamental dichotomy.}
In the LOCAL model, arguably the vast majority of research has focused on the class of so-called \emph{locally checkable problems}~\cite{naor1995what,brandt2019automatic}---problems that can be described via local constraints, such as coloring problems, which can be described by the local constraint that adjacent nodes receive different colors. 
Locally checkable problems contain most of the commonly studied problems, including maximal independent set, maximal matching, coloring problems, orientation problems, the algorithmic Lov\'asz local lemma problem (LLL), and many more.

By a seminal result of~\cite{chang2019exponential}, it is known that locally checkable problems satisfy a remarkable dichotomy: each problem either is solvable in $O(f(\Delta) \cdot \log^* n)$ deterministic rounds\footnote{$\log^* n$ is the minimum number of times the $\log$-function has to be applied iteratively to $n$ to obtain a value $\leq 1$.} for some function $f$ or exhibits a deterministic $\Omega(\log n)$-round complexity lower bound that already holds on constant-degree graphs.
In other words, the dependency of the deterministic complexity of any locally checkable problem on the size of the input instance is either very, very small or at least logarithmic, and both classes of problems have been studied in great detail over the past decades.
While the results of \cite{chang2019exponential} imply a ``normal form algorithm'' for solving any problem in the first category (such as maximal independent set, maximal matching, or $(\Delta + 1)$-coloring) in $O(f(\Delta) \cdot \log^* n)$ rounds, it is a major open question how to obtain algorithms with a runtime that depends\footnote{To be precise, whenever we talk about the dependency of the complexity of a problem (or of the runtime of an algorithm) on $n$, we mean the dependency on $n$ \emph{when admitting an arbitrary dependency on $\Delta$}. In other words, we are interested in functions $g(n)$ such that the considered complexity/runtime bounds are of the form $f(\Delta) \cdot g(n)$ for some function $f$.}
only logarithmically on $n$ for the problems in the second category.
Of course, many problems in this second category may not admit such fast algorithms at all, but there are a variety of important problems in this category, such as sinkless orientation, $\Delta$-coloring,
or LLL, for which algorithms with a polylogarithmic dependency on $n$ are known~\cite{ghaffari2017distributed,rozhon2020polylogarithmic}, making it plausible that a logarithmic dependency on $n$ might be achievable.

\paragraph{Logarithmic dependencies on $n$?}
Nevertheless, to the best of our knowledge, the only natural problems in the second category for which such a tight logarithmic dependency on $n$ is known are sinkless orientation~\cite{brandt2016lower} and closely related problems such as degree splitting problems~\cite{ghaffari2017distributed,ghaffari2020improved}.

Sinkless orientation is the problem of orienting the edges of the input graph such that each node of degree at least $3$ has at least one outgoing incident edge.
On a high level, a natural way to solve this problem is to find a suitable collection of non-overlapping cycles and nodes of degree at most $2$, orient, for each cycle $C$, the edges of $C$ consistently, and then orient all other edges towards the closest such cycle or node of degree at most $2$.
This guarantees that each node of degree at least $3$ has at least one outgoing edge and can be performed in $O(\log^2 n)$ rounds\footnote{See the discussion following \cref{thm:rulingsubgraphs} for more details.} in the LOCAL model.
Moreover, the underlying intuition---finding structures in the graph on which we can satisfy the local constraints specified by the studied problem (i.e., locally solve the problem) even if the output on the remaining graph is chosen adversarially---is sufficiently \emph{generic} that it arguably should also be applicable to other problems.

Unfortunately, the sinkless orientation algorithm that achieves a \emph{logarithmic} dependency on $n$~\cite{ghaffari2017distributed} makes use of a strong property that is specific to sinkless orientation---in particular the fact that a single outgoing edge suffices per node, which allows the authors to avoid having to consistently orient the aforementioned cycles.
As such, the approach from~\cite{ghaffari2017distributed} is considerably less generic than the $O(\log^2 n)$-round one and does not easily transfer to other problems.
Similarly, the approaches with a logarithmic runtime dependency on $n$ for the aforementioned splitting problems~\cite{ghaffari2017distributed,ghaffari2020improved} suffer from the same problem due to their close relation to sinkless orientation (e.g., the presented approaches often rely on sinkless orientation as a crucial subroutine).
This raises the following fundamental question.
\vspace{0.05cm}
\begin{tcolorbox}
	\textbf{Question 1}
    
    \noindent How can we develop a generic technique for designing deterministic algorithms exhibiting an only logarithmic runtime dependency on $n$?
\end{tcolorbox}

We remark that there is a close relation between deterministic and randomized complexities of locally checkable problems.
As shown by~\cite{chang2019exponential}, the deterministic and randomized complexities of any locally checkable problem differ at most exponentially, while the \emph{shattering} technique introduced in the distributed context by~\cite{barenboim2016locality} allows to achieve such an exponential improvement via randomness for a variety of problems.
As such, while we state some important questions only in the deterministic setting, there are closely connected randomized questions of a similar flavor.

\paragraph{$\Delta$-coloring.}
Perhaps the best-studied and most natural problem falling into the aforementioned second category (i.e., exhibiting an $\Omega(\log n)$ complexity lower bound) is the $\Delta$-coloring problem.
$\Delta$-coloring is a classic graph problem that has been studied in numerous computational models; the task is to color each vertex of the input graph with a color from $\{ 1, \dots, \Delta \}$ such that any two adjacent vertices receive different colors.
As opposed to the similarly well-studied problem of coloring with $\Delta + 1$ colors, there are graphs on which a proper $\Delta$-coloring does not exist; however, Brooks' Theorem~\cite{brooks1941colouring} provides a complete characterization of which graphs are $\Delta$-colorable: a graph $G$ can be properly $\Delta$-colored if and only if no connected component of $G$ is 1) a clique consisting of $\Delta+1$ nodes or 2) an odd cycle and $\Delta = 2$.
Consequently, when studying $\Delta$-coloring in an algorithmic context, the standard assumption is that any considered input graph $G$ has maximum degree $\Delta \geq 3$ and does not contain a clique of $\Delta+1$ nodes, i.e., $\Delta \geq 3$ and $G$ is $\Delta$-colorable.

In the LOCAL model, besides receiving plenty of attention in general, in particular in recent years \cite{panconesi95delta,brandt2016lower,chang2019exponential,ghaffari2021deterministic,fischer2023fast,balliu2022distributed}, $\Delta$-coloring has been tremendously important for the development of various research directions as we outline in the following.

\paragraph{The importance of $\Delta$-coloring.}
Many of the exciting developments of the last 8 years have their roots in two works published in 2016, which proved an $\Omega(\log \log n)$-round randomized~\cite{brandt2016lower} and an $\Omega(\log n)$-round deterministic~\cite{chang2019exponential} lower bound for $\Delta$-coloring (and sinkless orientation).
In~\cite{chang2019exponential}, the authors used these lower bounds to show the first (non-trivial) \emph{exponential separation} between the deterministic and randomized complexity of a locally checkable problem (namely, $\Delta$-coloring), and complemented this result by showing that a larger-than-exponential separation cannot exist for any locally checkable problem.
Moreover, the aforementioned lower bounds initiated a long line of research \cite{Grunau2022TheLO,Balliu2018AlmostGP,Balliu2017NewCO,brandt2016lower,chang2019exponential,ChangP19,fischer2017deterministic,ghaffari2021improved,ghaffari2018derandomizing,ghaffari2017distributed,linial1987distributive,Naor1991,linial1987distributive,rozhon2020polylogarithmic} studying the \emph{complexity landscape} of so-called \emph{LCL problems} (which are nothing else than locally checkable problems on constant-degree graphs), culminating in a complete classification of the possible deterministic complexities an LCL problem can exhibit, and an almost-complete classification for randomized complexities.

Perhaps even more importantly, the aforementioned randomized lower bound (on which the deterministic lower bound builds) was the first to be proved with a new lower bound technique called \emph{round elimination}, which since then has resulted in numerous lower bounds for many of the problems central to the LOCAL model \cite{brandt2016lower,ChangP19,chang2019distributed,brandt2019automatic,balliu2021lower,balliu2019hardness,brandt2020truly,binary,balliu2022ruling,balliu2021improved,balliu2022distributed}.
Fascinatingly, as shown in~\cite{brandt2021local}, another recent distributed lower bound technique, called \emph{Marks' technique}, based on ideas from the field of descriptive combinatorics~\cite{DetMarks,Marks_Coloring}, was also developed in the context of $\Delta$-coloring, providing the same deterministic $\Omega(\log n)$-round lower bound for $\Delta$-coloring that round elimination achieves in a different manner.
In fact, as stated in~\cite{brandt2021local}, it is an exciting open problem to find provable connections between round elimination and Marks' technique---something where $\Delta$-coloring, as one of the few natural problems to which both techniques apply, can be expected to play a prominent role.

However, $\Delta$-coloring has not only been instrumental for the development of new techniques, intriguingly it also appears \emph{structurally} in proofs for entirely different problems: as shown in~\cite{balliu2022distributed} (aptly called ``Distributed $\Delta$-Coloring Plays Hide-and-Seek''), a relaxation of $\Delta$-coloring used in~\cite{balliu2022distributed} to prove the aforementioned $\Omega(\log \log n)$-round randomized and $\Omega(\log n)$-round deterministic lower bounds for $\Delta$-coloring\footnote{To avoid confusion, we remark that even though~\cite{balliu2022distributed} and~\cite{brandt2016lower} both use round elimination to obtain the stated $\Delta$-coloring lower bounds, the approaches in the two papers are not the same and yield different proofs for the same bounds.} is also a crucial ingredient for obtaining the state-of-the-art lower bounds for fundamental problems like maximal independent set on trees, ruling sets or arbdefective colorings provided in the same work.
Currently, no way for obtaining these lower bounds that does not inherently make use of $\Delta$-coloring is known.
We remark that problems like maximal independent set and ruling sets fall into the first category in the discussed dichotomy (i.e., problems that can be solved in $O(f(\Delta) \cdot \log^* n)$ deterministic rounds for some function $f$), showing that, surprisingly, $\Delta$-coloring is highly relevant also for such problems, despite the fact that $\Delta$-coloring itself falls into the second category.

While the above discussion shows that $\Delta$-coloring has been one of the most influential problems for the development of the area of distributed graph algorithms in recent years, there is also little reason to believe that this trend will stop.
For instance, besides for the already mentioned open problem of the relation between round elimination and Marks' technique, $\Delta$-coloring is also an important problem for distributed complexity-theoretic questions:
Open Problem 7 in~\cite{balliu2022distributed} asks whether there exists a locally checkable problem that has a deterministic complexity of $\omega(\log n)$ and a randomized complexity of $O(\poly(\log n))$ on constant-degree graphs.
One of the natural candidates for such a problem is $\Delta$-coloring, with state-of-the-art bounds of $O(\log^2 n)$ and $\Omega(\log n)$ deterministic rounds and $O(\log^2 \log n)$ and $\Omega(\log \log n)$ randomized rounds previous to our work.
Moreover, Open Problem 7 (as stated above) could be answered in the negative by proving the Chang-Pettie conjecture~\cite[Conjecture 5.1]{ChangP19} stating\footnote{Technically, the Chang-Pettie conjecture states that the randomized complexity of a concrete problem---the LLL problem---is $O(\log \log n)$ (and its deterministic complexity $O(\log n)$, which again would follow by the separation result from~\cite{chang2019exponential}). However, as the LLL problem is a \emph{complete} problem for sublogarithmic randomized complexity (i.e., any LCL problem solvable in sublogarithmic randomized rounds has an asymptotic randomized complexity equal to or smaller than the complexity of the LLL problem), our phrasing of the conjecture is equivalent.} that any LCL problem with a sublogarithmic randomized complexity can be solved in $O(\log \log n)$ randomized rounds (which, by the aforementioned at-most-exponential-separation result from~\cite{chang2019exponential} provides the negative answer to Open Problem 7).
For proving the Chang-Pettie conjecture (which would resolve the last uncharted regime in the randomized LCL complexity landscape, completing this line of research spanning almost a decade), one natural obstacle is, again, $\Delta$-coloring.
Improving the state-of-the-art complexity bounds for $\Delta$-coloring on constant-degree graphs (or, equivalently, improving the dependency of the complexity bounds on $n$ on arbitrary graphs) would constitute a considerable (and necessary) step towards proving the Chang-Pettie conjecture.
We note that improving the best known bounds on the \emph{randomized} complexity of $\Delta$-coloring would require an improvement of the known bounds of the \emph{deterministic} complexity (due to the aforementioned separation result~\cite{chang2019exponential}), making both the randomized and the deterministic setting highly relevant.

In conclusion, $\Delta$-coloring is prototypical for the obstacles to making progress on various fundamental questions; hence obtaining a better understanding of (the complexity of) $\Delta$-coloring is of utmost importance.
In particular the following question is essential.
\vspace{0.05cm}
\begin{tcolorbox}
	\textbf{Question 2}
    
    \noindent What is the (randomized and deterministic) complexity of $\Delta$-coloring on constant-degree graphs, or, equivalently, what is the dependency of the complexity of $\Delta$-coloring on $n$ on arbitrary graphs?
\end{tcolorbox}

Unsurprisingly, answering this question will also be very useful for improving the state-of-the-art bounds for the complexity of $\Delta$-coloring purely as a function of $n$, without any dependencies on $\Delta$.

\subsection{Our Contributions}
\paragraph{A generic approach.}
As already hinted at in our discussion of sinkless orientation, a generic way\footnote{For more details, see \cref{sec:tech_overview}.} of deterministically solving problems like $\Delta$-coloring or sinkless orientation is, on an informal level, the following: compute a suitable (and not too ``sparse'') collection of non-overlapping substructures that are \emph{good} in the sense that they can be solved locally correctly (according to the local constraints of the problem) even if the output on the rest of the graph is adversarially chosen,
and then make use of the fact that all difficulties in solving the parts outside of the substructures can be ``pushed'' towards the substructures.
An essential technical ingredient of this general approach is the computation of a symmetry-breaking subroutine, called a ruling set, on $G^k$, the $k$th power of the input graph $G$, for some $k \in \Theta(\log n)$.
The ruling set can then be used to compute the aforementioned collection of substructures; however, already on constant-degree graphs, computing the ruling set itself takes $\Theta(\log^2 n)$ rounds, where one $\log$-factor comes from the time it takes to compute the ruling set \emph{if $G^k$ were the input graph}, and the other $\log$-factor from simulating this computation on $G^k$ on the actual input graph $G$.
In fact, as the authors of~\cite{balliu2022distributed} point out, the lower bound for ruling sets in their work ``implies that in order to improve\footnote{While this is not explicitly stated in the respective part of~\cite{balliu2022distributed} it is clear from the context that this is to be understood as improving the current best \emph{dependency on $n$}, or equivalently, the current best runtime on \emph{constant-degree} graphs.} the current best algorithm for $\Delta$-coloring~\cite{GHKM_dc21}, we need to find a genuinely different algorithm, which is not based on ruling sets.''
They even explicitly state as an open problem~\cite[Open Problem 8]{balliu2022distributed} to find a genuinely different algorithm that (at least) \emph{matches} the state-of-the-art complexity bound.

\paragraph{Ruling subgraph families.}
As our first contribution, we introduce the notion of \emph{ruling subgraph families} as a tool for solving problems from the aforementioned second category. 
Given a collection of subgraphs of the input graph, a ruling subgraph family is a subset of said collection such that the selected members are sufficiently far from each other and every non-selected member is close to a selected member (thereby generalizing the notion of ruling sets).
As one of our main technical contributions, we show how to compute such ruling subgraph families fast.
In fact, our main theorem about ruling subgraphs is somewhat more general than the above informal description suggests, allowing for a wider applicability. 

\begin{restatable}{theorem}{RulingSubgraphsThm}
    \label{thm:rulingsubgraphs}
    Let $\fR$ be a family of connected subgraphs of the input graph $G$, and let $k$ be an integer such that each member of $\fR$ has at most $k$ vertices. Let $t \leq k$ be an integer, and, for each member $H\in\fR$, let $X(H)$ be an induced connected subgraph of $H$ containing at most $t$ vertices.
    Let $d$ be an arbitrary nonnegative integer.
    Then there exists a deterministic algorithm running in $O((t+d) (d \log \Delta + \log k + \log^* n) + k \log^* n)$ rounds of LOCAL that computes a subfamily $\fR' \subseteq \fR$ such that
    \begin{itemize}
\item for every two distinct subgraphs $H_1',H'_2\in \fR'$, the distance between $H_1'$ and $X(H_2')$ is at least $d + 1$ and
        \item for every subgraph $H \in \fR \setminus \fR'$, there is a subgraph $H' \in \fR'$ such that the distance between $H$ and $H'$ is in $O((t+d)(d \log \Delta + \log k) + k \log^* k)$.\end{itemize}
\end{restatable}

Here, as usual, the distance between two subgraphs $H, H'$ of a graph $G$ is defined as the minimum distance between a vertex of $H$ and a vertex of $H'$, taken over all pairs of vertices from $V(H) \times V(H')$.

By setting, for instance, $d := 1$ and choosing $k$ and $t$ suitably from $\Theta(\log_{\Delta} n)$, \cref{thm:rulingsubgraphs} asserts that a ruling subgraph family (with the respective parameters) can be computed in $\Otilde(\log n)$ rounds.\footnote{We will use $\Otilde (\cdot)$ to hide $\log \log n$-factors in deterministic algorithms and $\log \log \log n$-factors in randomized algorithms.}
This constitutes a considerable step towards answering Question 1, by providing a method to perform the generic approach outlined for $\Delta$-coloring and sinkless orientation without resorting to the costly computation of a ruling set on a power graph.
In fact, for sinkless orientation, \cref{thm:rulingsubgraphs} directly yields an algorithm\footnote{Informally, start by letting each node $v$ select $3$ incident edges and disconnecting $v$ from all other incident edges (adding a new virtual endpoint for the edge instead of $v$). On the obtained graph (of maximum degree $3$), let each node select a cycle of length $O(\log n)$ or a node of degree $\leq 2$ in its $O(\log n)$-hop neighborhood (one of the two must exist). Then apply \cref{thm:rulingsubgraphs} on the family of selected objects, orient the selected cycles consistently, and orient all edges not contained in any of the selected cycles towards the closest selected object.} running in $\Otilde (\log n)$ rounds, providing evidence for our claim that the outlined approach is generic.
As our next contribution (providing further evidence), we use the developed tool to improve on the state-of-the-art bounds for the complexity of $\Delta$-coloring.

\paragraph{Deterministic $\Delta$-coloring in LOCAL.}
First, we improve the complexity in the deterministic LOCAL model, parameterized by $n$ and $\Delta$.
\begin{restatable}{theorem}{deterministicLOCAL}
    \label{thm:detLOCAL}
        There exists an $O(\log^4 \Delta + \log^2\Delta \log n \log^* n)$-round algorithm for $\Delta$-coloring graphs of maximum degree $\Delta \geq 3$ in the deterministic LOCAL model.
\end{restatable}
While, at a superficial glance, this might appear as a modest improvement over the previous state of the art of $O(\log^2 \Delta \log^2 n)$ rounds~\cite{ghaffari2021deterministic}, \cref{thm:detLOCAL} is in fact fundamentally important from several perspectives.

First of all, it almost resolves the deterministic LOCAL complexity of $\Delta$-coloring on constant-degree graphs, improving
the previous state-of-the-art bound of $O(\log^2 n)$ rounds~\cite{GHKM_dc21} to $O(\log n \log^* n)$, which, due to the known lower bound of $\Omega(\log n)$ rounds~\cite{chang2019exponential}, is tight up to the razor-thin $\log^* n$-factor.
It seems unlikely that minor tweaks to our algorithm could remove
this $\log^* n$-term, as at its core is some symmetry breaking between $\Theta(\log n)$-large structures.
Achieving this specific complexity in this way reframes the question of the complexity of $\Delta$-coloring in a new light: could it be that such symmetry breaking between large structures is essential to $\Delta$-coloring, and $\Theta(\log n \log^* n)$ is its complexity on constant-degree graphs, or does there exist a more efficient algorithm that circumvents the limitations inherent to our approach?

\Cref{thm:detLOCAL} also results in more subtle improvements  over previous and current state-of-the-art in larger regimes of $\Delta$ (see \cref{sec:high-delta-optimization}).

We also would like to point out that while ruling subgraph families do have obvious similarities to ruling sets, the approach based on ruling subgraph families is novel and therefore can arguably be seen as answering Open Problem 8 from~\cite{balliu2022distributed}.

\paragraph{Randomized $\Delta$-coloring in LOCAL.}
Using the same technique, we also obtain improvements in the randomized LOCAL model, where we obtain an $O(\sqrt{\Delta \log\Delta} (\log_\Delta \log n \cdot \log^* n + \log \log n))$-round algorithm, improving on a previous complexity bound of $O(\sqrt{\Delta \log\Delta} \log^2 \log n )$ \cite{GHKM_dc21}.

\begin{restatable}{theorem}{randomizedLOCAL}
    \label{thm:randLOCAL}
    There exist randomized LOCAL algorithms for $\Delta$-coloring graphs of maximum degree $\Delta \geq 3$ w.h.p.\ with complexities $O(\sqrt{\Delta \log\Delta} (\log_\Delta \log n \cdot \log^* n + \log \log n))$ or $\Otilde (\log^{8/3} \log n)$.
\end{restatable}

In the $\Delta \in O(1)$ regime, our algorithm achieves an almost-optimal complexity of $O(\log \log n \log^* n)$ compared to the lower bound which is $\Omega(\log \log n)$.
This makes progress on the Chang-Pettie conjecture and on Open Problem 7 from~\cite{balliu2022distributed} (together with \cref{thm:detLOCAL}), by almost ruling out $\Delta$-coloring as a (counterexample) candidate.
Moreover, \cref{thm:detLOCAL} and \cref{thm:randLOCAL} make substantial progress towards answering Question 2, posed above.

We remark that while the gap between the best known randomized complexities in LOCAL and CONGEST remains substantial, a recent result provided an $O(\poly(\log \log n))$ algorithm for $\Delta$-coloring in CONGEST \cite{HM_disc24}. For $\Delta$ large enough ($\Delta \in \Omega(\log^{21} n)$), the complexity is known to collapse to $O(\log^* n)$ in both LOCAL and CONGEST~\cite{fischer2023fast}.

\subsection{Organization of the Paper}
\label{sec:organization}

We give an overview of all our results and the key ideas behind them in \cref{sec:tech_overview}, with detailed proofs and formal statements in later sections.

Our results on computing ruling subgraph families (including \cref{thm:rulingsubgraphs}) are the topic of \cref{sec:ruling-subgraphs}. In the subsequent \cref{sec:detresults}, we give our main application of the technique, a faster deterministic algorithm for $\Delta$-coloring in \LOCAL (\cref{thm:detLOCAL}). The implications for $\Delta$-coloring in randomized \LOCAL are presented in \cref{sec:randomized} (\cref{thm:randLOCAL}).

We lay out the notation we use throughout the paper and important definitions in \cref{sec:preliminaries}.
\Cref{sec:complexity-in-n,sec:sota-congest} provide an exposition of the state-of-the-art.
\Cref{sec:nldcc-proof} contains some proofs deferred from the main text.

\section{Preliminaries and Notation}
\label{sec:preliminaries}

We will use the notation $[k] = \set{1,\dots,k}$. We represent a current partial coloring by the function $\col: V \to [\Delta+1] \cup \set{\bot}$. $\col(v)$ is the current color of node $v$; $\bot$ represents the uncolored state.

Throughout the paper, $G=(V,E)$ refers to both the communication graph and the graph to $\Delta$-color, of vertex set $V$ and edge set $E$. $n$ is the number of nodes of $G$, and $\Delta$ its maximum degree, both assumed to be known by all nodes.
Throughout the paper\footnote{The results from \cref{sec:ruling-subgraphs} still apply when $\Delta < 3$, but we are not aware of any interesting applications of its results with $\Delta \in \set{1,2}$. $\Delta$-coloring requires $\Delta \geq 3$ to guarantee the existence of useful subgraphs within moderate distance from each node.},
we assume $\Delta \geq 3$.
When considering another graph $H$, we denote by $V(H)$, $E(H)$, and $\Delta(H)$, its vertex set, edge set, and maximum degree.

We denote by $N(v)$ the neighborhood in $G$ of a node $v \in V$. When considering a node $v \in V(H)$ in a graph $H \neq G$, we denote by $N_H(v)$ its neighborhood in $H$. For two nodes $v,v'$ in a graph $H$ (possibly $H=G$), their distance in $H$, denoted by $\dist_H(v,v')$, is the length of the shortest path in $H$ between $v$ and $v'$ ($+\infty$ if no such path exists), that is, the smallest $k$ s.t.\ there exist $k+1$ vertices $v_0,v_1,\ldots,v_k$ with $v_0 = v$ and $v_k = v'$ such that $v_i \in N_H(v_{i+1})$ for each $i \in [k]$. For any positive integer $k$, the power graph $G^k$ is defined as $V(G^k) = V(G)$ and $E(G^k) = \set{uv: u \in V(G), v \in V(G) \setminus \set{u}, \textrm{s.t.\ } \dist_G(u,v) \leq k}$. The distance-$r$ neighborhood of $v$ is $N^r(v) = N_{G^r}(v)$. For two sets of nodes $S,S'\subseteq V$, we define their distance as $\dist_G(S,S') = \min_{v\in S, v'\in S'} \dist_G(v,v')$. We similarly define $\dist_G(H,H') = \dist_G(V(H),V(H'))$ for two subgraphs $H,H'$ of $G$. For a subset $S \subseteq V$, the subgraph induced by $S$ in $G$, denoted by $G[S]$, is the graph with $S$ as vertex set and the edge set $\set{uv \in E: u\in S \wedge v \in S}$. $H$ is an induced subgraph of $G$ if $H = G[V(H)]$.

We will use the following results from prior work.

\paragraph{Structural graph theoretical results.}

\begin{definition}[Degree-choosable component]
\label{def:dcc}
    A graph is \emph{degree-choosable} if, given a list $L(v)$ of colors with $|L(v)| \geq \deg(v)$ for every node $v$, it is always possible to properly color the graph s.t.\ each node $v$ uses a color from $L(v)$.
    A subgraph $G'$ of a graph $G$ is called a \emph{degree-choosable component (DCC)} if $G'$ is a $2$-connected induced subgraph that is degree-choosable.
\end{definition}

\begin{definition}[$\Delta$-extendable]
    \label{def:delta-extendable}
    An induced subgraph $H$ of $G$ is \emph{$\Delta$-extendable} if for any partial coloring of the nodes $V(G) \setminus V(H)$, the partial coloring can always be extended to the nodes of $H$ using only colors from $[\Delta]$.
    DCCs and nodes of degree $< \Delta$ are $\Delta$-extendable.
\end{definition}

$2$-connected graphs in particular include odd cycles and cliques of size at least $2$, which are not degree-choosable. All other types of $2$-connected graphs are degree-choosable, including even cycles.
There is a natural decomposition of any given connected graph into its maximal $2$-connected subgraphs.

\begin{definition}[Blocks and block graph]
\label{def:block-graph}
    The \emph{blocks} of a graph $G=(V,E)$ are its maximal $2$-connected subgraphs. Since blocks are necessarily induced subgraphs by their maximality, we also refer to their vertex sets as being blocks. 
    
    Let $B\subseteq 2^V$ be the blocks of a graph $G =(V,E)$, the \emph{block graph} of $G$ is the bipartite graph $\fB(G)=(B \sqcup V,E')$ where for each block $S \in B$ and vertex $v\in V$ of the original graph $G$, $vS \in E'$ iff $v\in S$. 
\end{definition}

\begin{proposition}[See e.g.\ {\cite[Section 3.1]{Diestel_GT_5th_ed}}]
\label{prop:block-tree}
    The block graph of a connected graph is a tree.
\end{proposition}

This property of block graphs underpins the name of Gallai ``tree'' in the next definition. 

\begin{definition}
\label{def:gallai-tree}
    A \emph{Gallai tree} is a connected graph whose blocks are each a clique or an odd cycle.
\end{definition}

\begin{lemma}[\cite{KSW_dm96,Thomassen_jct97a}]
    A connected graph is degree-choosable if and only if it is not a Gallai tree.
\end{lemma}

\begin{lemma}[{\cite[Lemma 9]{GHKM_dc21}}]
    \label{lem:exists-dcc}
    In a graph $G$ of maximum degree $\Delta \geq 3$, the distance-$2 \log_{\Delta-1} n$ neighborhood of every node $v$ contains a node of degree less than $\Delta$ or a degree-choosable component.
\end{lemma}

\begin{lemma}[{\cite[Lemma 8]{GHKM_dc21}}]
    \label{lem:DCCorexpand}
    For every node $v$ in a graph $G$, there exist either a node of degree less than $\Delta$ in distance $k$ of $v$ or a degree choosable component of radius at most $k$ in distance $k$ of $v$, or there are at least $(\Delta-1)^{\floor{k/2}}$ nodes at distance $k$ of $v$.
\end{lemma}

We make use of more specific subgraphs than DCCs for our $\Delta$-coloring algorithms, which we call \emph{locally degree-choosable components} and \emph{nice} locally degree-choosable components. We show a lemma analogous to \cref{lem:exists-dcc} about their existence in \cref{sec:nice-ldccs}.

\begin{definition}[Locally degree-choosable component]
\label{def:ldcc}
    A DCC is a \emph{locally degree-choosable component (LDCC)} if it is not an even cycle.
\end{definition}

\begin{restatable}{definition}{defNLDCC}[Nice locally degree-choosable component]
    \label{def:nice-ldcc}
    An LDCC is a \emph{nice locally degree-choosable component} (NLDCC) if it has exactly two nodes of degree $3$, while all its other nodes have degree $2$.
\end{restatable}

\begin{definition}[Nice locally $\Delta$-extendable components]
    \label{def:nice-extendable}
     NLDCCs and nodes of degree less than $\Delta$ are nice locally $\Delta$-extendable components (\NLDelEC).
\end{definition}

\paragraph{Algorithmic subroutines relevant for our algorithms.}

\begin{lemma}[Linial's algorithm, Theorem 4.2 in \cite{linial1992locality}]
    \label{lem:linial}
    There is an algorithm for computing a $O(\Delta^2)$-coloring of a graph $G$ in $O(\log^*n)$ rounds.
\end{lemma}

A subset $S$ of the nodes of a graph $G$ is an $(\alpha,\beta)$-ruling set~\cite{ColeV86,Awerbuch89} of $G$ iff the following two properties hold: any two nodes in $S$ are at distance at least $\alpha$ from each other, and each node in the graph $G$ is within distance $\beta$ from some node in $S$. Ruling sets generalize maximal independent sets (MIS), which are $(2,1)$-ruling sets.

\begin{lemma}[Computing ruling sets]
    \label{lem:ruling-set}
    There exist algorithms in \LOCAL for the following tasks
    \begin{itemize}
        \item Deterministically, given a $d$-coloring of a graph $G$, computing a $(2,c)$-ruling set of $G$ in $O( c \cdot d^{1/c})$ rounds~\cite[Theorem 3]{schneider2013symmetry}.
        \item Deterministically, computing a $(2,O(\log \log \Delta))$-ruling set in $\Otilde(\log n)$ rounds~\cite[Theorem 4.1 in arxiv version]{GG_focs24}.
        \item With high probability (randomized), computing a $(2,O(\log \log n))$-ruling set in $O(\log \log n)$ rounds \cite{Gfeller07,schneider2013symmetry}.
    \end{itemize}
\end{lemma}

In particular, starting from an $O(\Delta^2)$-coloring, a $(2,\log \Delta)$-ruling set can be computed in $O(\log \Delta)$ rounds.

\begin{lemma}[Computing MIS]
    \label{lem:mis_local}
    An MIS can be computed deterministically in the following runtimes in \LOCAL:
    \begin{itemize}
        \item $O(\Delta + \log^*n)$ rounds~\cite{barenboim14distributed}.\item $O(\log^2 \Delta \log n)$~\cite{faour2023local}.
        \item $\Otilde(\log^{5/3} n)$~\cite{GG_focs24}.
    \end{itemize}
\end{lemma}

\begin{lemma}
\label{lem:d1lc_local}
    A $\deg+1$-list-coloring with colors of order $O(\Delta)$ can be computed in the following runtimes in \LOCAL:
\begin{itemize}
    \item Deterministically, in $O(\sqrt{\Delta \log \Delta} + \log^* n)$ rounds \cite[Theorem 3]{MausTonoyan_dc22}, following work by \cite{fraigniaud2016local,BEG_jacm22}. The algorithm only takes $O(\sqrt{\Delta \log \Delta})$ rounds if given a $\Delta^{O(1)}$-coloring.
    \item Deterministically, in $O(\log^2 \Delta \log n)$ rounds \cite{ghaffari2021deterministic}
\item Deterministically, in $\Otilde(\log^{5/3} n)$ rounds \cite[Theorem 1.3 in arxiv version]{GG_focs24}
\item With high probability (randomized), in $\Otilde(\log^{5/3} \log n)$ rounds \cite{HalldorssonKNT22near,GG_focs24} 
\end{itemize}
\end{lemma}

\begin{definition}[Layered graph]
    \label{def:layered_graph}
    A \emph{layered graph} of height $h$ is a graph in which the set of nodes $V$ is partitioned into $h$ subsets (layers), $V = \bigsqcup_{i=1}^h V_i$. For any node $v$ in some layer $V_i$, its \emph{up degree} $\widehat{\delta}(v)$ is the number of neighbors $v$ has in layers of equal or lesser index, i.e., $\card{N(v) \cap \bigsqcup_{j \leq i} V_j}$. We denote by $\widehat{\Delta}$ any known global bound on the maximum up degree $\max_v \widehat{\delta}(v)$.
\end{definition}

Each algorithm of complexity $T(n,\Delta)$ for coloring general graphs of size $n$ and maximum degree $\Delta$ mentioned in \cref{lem:d1lc_local} implies an algorithm of complexity $O(h \cdot T(n,\widehat{\Delta}))$ for coloring layered graphs of height $h$, size $n$ and maximum up degree $\widehat{\Delta}$, by coloring layers in order of decreasing index.
We will also use the following result that is specifically about coloring layered graphs:

\begin{lemma}[{\cite[Lemma 5.4 in arxiv version]{ghaffari2021deterministic}}]
\label{lem:layered_graphs}
    Consider a layered graph where each node $v \in V$ receives a list of colors $L_v$ of size $|L(v)| \geq \widehat{\delta}(v)+1$. There exists a deterministic \LOCAL algorithm that list-colors the graph in
    $O(\log^2 \widehat{\Delta} \cdot \log n + h \cdot \log^3 \widehat{\Delta})$ rounds.
\end{lemma}

\section{Technical Overview}
\label{sec:tech_overview}

\paragraph{Approaches from prior work: deterministic setting}
In the deterministic \LOCAL model, the general approach to computing $\Delta$-coloring in prior work~\cite{panconesi95delta,GHKM_dc21,ghaffari2021deterministic} has been a three step process of computing a ruling set $S$, coloring nodes outside $S$ in order of decreasing distance from $S$, and finally coloring each node in $S$ by finding a nearby subgraph with enough flexibility to color it using only colors from $[\Delta]$ regardless of colors taken outside the subgraph.

Underpinning this approach are two important ideas. First, that when coloring nodes outside $S$, each node has at least one neighbor closer to $S$, and thus at most $\Delta-1$ neighbors that are at higher or equal distance from $S$. This makes it easy to color nodes outside $S$ without the color $\Delta+1$ by treating them in order of decreasing distance from $S$. Second, any graph without $\Delta$-extendable subgraphs (i.e., subgraphs that can be colored last without using color $\Delta+1$) necessarily \emph{expands}, i.e., contains $\Delta^{\Omega(r)}$ nodes at distance $r$ from each node. Therefore, an easy-to-recolor subgraph needs to be within distance $O(\log_\Delta n)$ of each ruling set node. Having a sufficient $\Omega(\log_\Delta n)$ distance between nodes from the ruling set allows them to identify such a subgraph and recolor themselves in parallel, without interfering with one another.

\paragraph{Approaches from prior work: randomized setting}
As many other problems in distributed computing, the state-of-the-art randomized complexity of $\Delta$-coloring is almost exponentially smaller than its state-of-the-art deterministic complexity. An intuitive reason for this speedup is that in $\Delta$-coloring as in other problems, solving an $n$-sized instance in the randomized setting can essentially be reduced to solving $O(\poly(\log n))$-sized instances deterministically, using a technique known as \emph{shattering} \cite{barenboim2016locality}.

In the case of $\Delta$-coloring, this drastic reduction of the size of the largest connected component of unsolved nodes is achieved by the creation of ``havens'' (formally called $T$-node) in the expanding parts of the graph, through random sampling. Also used in the analysis of \cite{panconesi95delta}, $T$-nodes in \cite{GHKM_dc21} are randomly sampled across the graph, maintaining some minimal constant distance between them while making their sampling likely enough that any large subgraph of $G$ contains a large number of $T$-nodes. As alluded to when we referred to them as havens, the role of $T$-nodes is to provide locations in the graphs where the $\Delta$-coloring can be completed last. This is achieved by having each $T$-node select two random neighbors which it colors with the same color. This repetition of a color around each $T$-node means that no more than $\Delta-1$ colors will be used in each $T$-node's neighborhood as we complete the coloring outside of them, and thus, $T$-nodes are always guaranteed to be able to color themselves with a color in $[\Delta]$. In a careful analysis of this process, Ghaffari, Hirvonen, Kuhn, and Maus~\cite{GHKM_dc21} proved that in parts of the graph that expand, every node is guaranteed to have an $O(\log_\Delta \log n)$ uncolored path to a $T$-node. Conversely, parts of the graph with less expansion are guaranteed to contain a $\Delta$-extendable subgraph of diameter $O(\log_\Delta \log n)$. Put together, those two observations allow to reduce $\Delta$-coloring to a few instances of degree+1-list coloring. Prior work, due to their use of ruling sets to break symmetries between $\Delta$-extendable subgraphs, ended with an algorithm of complexity $O(\log^2 \log n)$ on constant degree graph.

Random sampling plays a similarly crucial role 
in the state-of-the-art (and tight), $O(\log_\Delta \log n)$ algorithm for sinkless orientation \cite{ghaffari2017distributed}.

We note that some recent papers~\cite{fischer2023fast,HM_disc24} deviate from this approach, relying instead on graph decomposition results known as \emph{almost-clique decompositions} (ACD). At a high level, an almost-clique decomposition is a partition of the nodes of the graph instead sparse parts, i.e., with vertices whose neighborhoods contain significantly less than $\binom{\Delta}{2}$ edges, and highly-connected clusters called \emph{almost-cliques}.
Approaching $\Delta$-coloring through the lens of almost-clique decompositions has also appeared in the streaming setting~\cite{AssadiKM_theoretics23}.

Our own approach is in the continuity of the sequence of works which did not use almost-clique decompositions \cite{panconesi95delta,GHKM_dc21,ghaffari2021deterministic}.

\paragraph{Main technical idea: Ruling Subgraphs}
The previous approach relying on a ruling set
with large distance between its members 
so as to recolor them independently currently incurs a cost of $\Theta(\log^2 n / \log \Delta)$ for computing said ruling set. Improving on the current state-of-the-art $O(\log^2 n)$ deterministic complexity in \LOCAL on constant degree graphs requires one of three things: improving the complexity of computing ruling sets, finding a way of efficiently completing the coloring even if ruling set nodes are closer to each other, or taking a different approach that does not rely on ruling sets in the same manner. Our new algorithm takes this third option.

When relying on coloring nodes from a ruling set last, its nodes are taken to be at a sufficient $\Omega(\log_\Delta n)$ distance from each other to ensure the coloring can be completed on each of them independently of other nodes in the ruling set. 
A critical observation is that such a large distance between ruling set nodes is chosen so as to guarantee that subgraphs used by ruling set nodes to recolor themselves \emph{neither touch nor overlap}. 
We take this idea one step further by computing \emph{ruling subgraph families} instead of ruling sets. That is, we compute a set of non-interfering $\Delta$-extendable subgraphs, each of size at most $O(\log_\Delta n)$, such that each node in the graph is within distance $\Otilde(\log n) \ll \Theta(\log^2 n)$ of one of them. Organizing nodes outside those subgraphs according to their distance to the closest such subgraph thus defines less layers than in previous algorithms, and allows us to compute a coloring faster, saving an $\Otilde(\log n)$ factor.

To understand the difference between the previous approach and our own, it can be useful to cast the earlier algorithm in our paradigm of ruling subgraphs.
Ensuring a distance of $\Omega(\log_\Delta n)$ between nodes in a ruling set is equivalent to breaking symmetry between balls of radius-$\Omega(\log_\Delta n)$ in our paradigm.
We improve over this by breaking symmetry between sparser structures -- not just of diameter $O(\log_\Delta n)$, but also containing only $O(\log_\Delta n)$ nodes.
Our algorithm can also be understood as reversing two steps in the prior approach. Instead of breaking symmetry between large balls, before looking for useful structures within them that allows to color with only $\Delta$ colors, we instead first look for such structures, and break symmetry between them.

An important ingredient of our algorithm for computing ruling subgraph families are subroutines we call \emph{elimination walks}. At a high level, our algorithm starts by discovering all locally $\Delta$-extendable subgraphs of size $O(\log_\Delta n)$ -- subgraphs s.t.\ any partial $\Delta$-coloring outside them can be extended to them.
We gradually eliminate many subgraphs of this family until the survivors no longer interfere with one another. 
Considering the virtual graph whose nodes are members of a current subgraph family, and whose edges represent adjacency between said subgraphs, we give an $O(\log^* n)$-phase procedure, where each phase reduces the maximum degree of the virtual graph to a polylogarithm of its previous value. Each phase consists of two subroutines, one of which is an elimination walk.
The technique has each subgraph introduce a token, which performs a walk on its nodes, sometimes making a detour when it gets close to another subgraph. Whenever multiple tokens are on the same node during the walk, all but one of them is eliminated. Coupled with another subroutine based on ruling sets, this simple technique is able to achieve dramatic degree reductions in the virtual graph in a short amount of time.

The computing of ruling subgraph families is presented in \cref{sec:ruling-subgraphs}, and the overall deterministic \LOCAL algorithm in \cref{sec:detlocal}.

\paragraph{Consequences in the randomized setting}

As occurs frequently in the theory of distributed algorithms, due to the existence of techniques like shattering connecting the complexity in the randomized setting to that of the deterministic setting, our deterministic algorithm also improves the randomized complexity of $\Delta$-coloring. We explain the framework of prior work for randomized $\Delta$-coloring and what results from plugging our deterministic algorithms into it in \cref{sec:randomized}.

\section{Ruling Subgraph Family}
\label{sec:ruling-subgraphs}

In this section, we show how to compute a generalization\footnote{In fact, setting $t=k=1$ and $\fR = \set{\set{v}:v \in V(G)}$ in \cref{thm:rulingsubgraphs} results in the computation of a $(d+1,O(d^2 \log\Delta))$-ruling set of $G$.} of ruling sets, where the ruling/ruled objects are (connected) \emph{subgraphs} instead of nodes, efficiently.
Concretely, the objective of this section is to prove \cref{thm:rulingsubgraphs}, which we recall for convenience.

\RulingSubgraphsThm*

We note that if the family $\fR$ is sufficiently dense in $G$, then every node in $G$ also has a member of $\fR'$ nearby.
More precisely, we obtain the following simple corollary of~\cref{thm:rulingsubgraphs}.

\begin{corollary}
    \label{cor:short-distance}
    Let $\fR$ and $\fR'$ be as in \cref{thm:rulingsubgraphs} and $b$ a nonnegative integer.
    If every node of the input graph $G$ is within distance $b$ of a member of the subgraph family $\fR$, then every node of $G$ is within distance $b + O((t+d)(d \log \Delta + \log k) + k \log^* k)$ from a member of the family $\fR'$.
\end{corollary}

By choosing $X(H) = H$ for each $H \in \fR$ in~\cref{thm:rulingsubgraphs}, we also obtain the following corollary which, compared to~\cref{thm:rulingsubgraphs}, provides a slightly less powerful but easier-to-understand statement about the computation of a generalization of ruling sets to subgraphs (instead of nodes).

\begin{corollary}
    \label{cor:rulingsubgraphPrev}
    Let $\fR$ be a family of connected subgraphs of the input graph $G$, and let $k$ be an integer such that each member of $\fR$ has at most $k$ vertices.
    Let $d$ be an arbitrary nonnegative integer.
    Then there exists a deterministic algorithm running in $O((k+d) (d \log \Delta + \log k + \log^* n))$ rounds of LOCAL that computes a subfamily $\fR' \subseteq \fR$ such that
    \begin{itemize}
        \item the distance between any two subgraphs in $\fR'$ is at least $d + 1$, and
        \item for every subgraph $H \in \fR \setminus \fR'$, there is a subgraph $H' \in \fR'$ such that the distance between $H$ and $H'$ is in $O((k+d)(d\log \Delta + \log k))$.
    \end{itemize}
\end{corollary}

Before proving~\cref{thm:rulingsubgraphs}, we collect some useful definitions.
We start by defining a special type of subgraph that we will rely on heavily in the proof of~\cref{thm:rulingsubgraphs}.
\begin{definition}[Orchid]
\label{def:orchid}
Let $t$ be a positive and $d$ a nonnegative integer. 
A \emph{$(t,d)$-orchid} $F(H)$ of a subgraph $H$ of the input graph $G$ is a connected induced subgraph of $G$ that contains
\begin{itemize}
\item a designated node of $H$ called the \emph{root} and
\item a connected induced subgraph of $H$ called the \emph{stem} such that
\begin{itemize}
    \item the stem consists of at most $t$ nodes,
    \item the stem contains the root as a node, and
    \item $F(H)$ is the induced graph whose node set consists precisely of all nodes at distance at most $d$ of the stem (in $G$).
\end{itemize}
\end{itemize}
In particular, while the root and the stem of $F(H)$ are part of $H$, the orchid $F(H)$ itself is not necessarily a subgraph of $H$.
We may omit the parameters $t,d$, and just refer to orchids when the parameters are clear from the context (or irrelevant).
While for a given orchid, there might be more than one choice for the root and the stem, we assume that each orchid comes with a fixed root and stem.
Moreover, we will frequently consider a family $\fR$ of subgraphs; when considering $(t,d)$-orchids in this context, we will implicitly assume that for each graph $H \in \fR$ a $(t,d)$-orchid $F(H)$ of $H$ is fixed.
(Note that from the description of a $(t,d)$-orchid it follows directly that for any choice of $t$, $d$, and $H$, a $(t,d)$-orchid $F(H)$ exists.)
\end{definition}

\begin{figure}[ht]
    \centering
\begin{tabular}{*{3}{c}}
\includegraphics[width=0.25\linewidth,page=1]{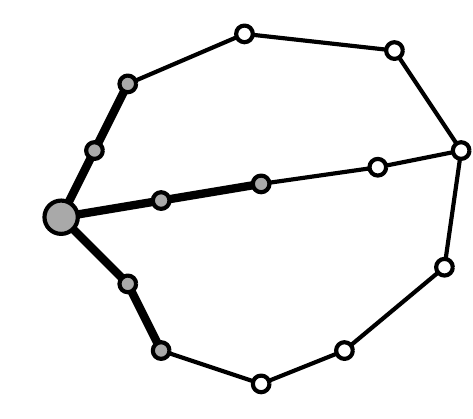}
\hspace{0.05\linewidth}\includegraphics[width=0.25\linewidth,page=2]{stem_orchid_branches_example_v3.pdf}
\hspace{0.05\linewidth}\includegraphics[width=0.25\linewidth,page=3]{stem_orchid_branches_example_v3.pdf}
\end{tabular}
\caption{Three representations of $H$ and its orchid. Left: the graph with its stem highlighted and its root node represented as the bigger node. Middle: the orchid is a zone surrounding and including the stem. Right: the stem with highlighted \emph{branch nodes}, which we introduce later (\cref{def:branch-nodes}).} 
    \label{fig:stem_orchid_branches}
\end{figure}

Next, we define the notion of a distance-$d$ conflict graph. 

\begin{definition}[distance-$d$ conflict graph]
    For a family $\fR$ of subgraphs $H_1,\ldots,H_{\card{\fR}}$ of a graph $G$, the \emph{distance-$d$ conflict graph (of $\fR$)} is the graph $C_\fR$ with node set $V(C_\fR) := \fR$ and edge set $E(C_\fR) := \set{\set{H,H'} \mid H, H' \in \fR, H \neq H', \text{ and } \dist_G(H,H') \leq d}$.
\end{definition}

Note that we do not require that the subgraphs $H_i$ are disjoint.
In particular, the distance-$0$ conflict graph (which connects subgraphs at distance $0$ from each other, i.e., intersecting subgraphs) does not necessarily consist of isolated nodes.

We also define the notion of a contention digraph for our families of special subgraphs with orchids.

\begin{definition}[Contention digraph]
    Consider a family of subgraphs $\fR$ in some graph $G$, and assume that for each member $H \in \fR$ (and parameters $t$, $d$) a $(t,d)$-orchid $F(H)$ of $H$ has been fixed. We denote by $C(\fR,G)$ the \emph{contention digraph (of $G$)}, defined via
    \begin{enumerate}
        \item $V(C(\fR,G)) = \fR$ and
        \item $E(C(\fR,G)) = \{ (H, H') \mid H, H' \in \fR, H \neq H', \text{ and } V(H)\cap V(F(H')) \neq \emptyset \}$.
\end{enumerate}
    In other words, for a family of subgraphs $\fR$, whenever a subgraph $H\in \fR$ overlaps with the (fixed) orchid  $F(H')$ of another subgraph $H'\in \fR$, the directed edge $(H,H')$ exists in the contention digraph $C(\fR,G)$.

    We denote by $\indeg_{\max}(\fR,G)$ (resp.\ $\outdeg_{\max}(\fR,G)$) the maximum indegree (resp.\ outdegree) of the contention digraph of a family $\fR$ in some graph $G$. When clear from context, we omit $G$ from the notation.  
\end{definition}

Throughout the remainder of the section, assume that the parameters $t$, $d$, and $k$ appearing in~\cref{thm:rulingsubgraphs} have been fixed, and that for each member $H \in \fR$, a $(t,d)$-orchid $F(H)$ has been fixed.

We now turn our attention towards proving~\cref{thm:rulingsubgraphs}.
To this end, we fix as $(t,d)$-orchid for each subgraph $H$ the one whose stem is the set $X(H)$ mentioned in~\cref{thm:rulingsubgraphs}. 
Thereby, in order to obtain the property stated in the first bullet of~\cref{thm:rulingsubgraphs}, it suffices to ensure that no subgraph contained in $\fR$ nontrivially intersects the orchid of a different subgraph contained in $\fR$.
Observe that this sufficient condition is equivalent to the statement that the contention digraph $C(\fR,G)$ does not contain any edges.
As such, on a high level our overall approach to proving~\cref{thm:rulingsubgraphs} can be summarized as follows: iteratively decrease the degree of the nodes in the contention digraph by removing subgraphs from $\fR$ while simultaneously making sure that we do not remove too many subgraphs from $\fR$ to not violate the property stated in the second bullet of~\cref{thm:rulingsubgraphs}.

The algorithm achieving \cref{thm:rulingsubgraphs} is described in \cref{alg:rulingsubgraphs} (\RulingSubgraphs).
It proceeds in four steps that, roughly speaking, can be described as follows.
First, we remove subgraphs from $\fR$ in a way that ensures that, for each orchid of a remaining subgraph, its $d$-hop neighborhood does not intersect too many orchids of other remaining subgraphs.
This enables the efficient computation of a suitable ruling set on the distance-$d$ conflict graph in the second step (as the degree of this conflict graph is sufficiently small due to the first step).
We show that, by removing the subgraphs not corresponding to the computed ruling set nodes from (the remainder of) $\fR$, we decrease the maximum outdegree of the contention digraph to $O(k/d)$.
Similarly, in the third step, we reduce the maximum indegree of the contention digraph to $O(k\cdot t)$ by removing further subgraphs from (the remainder of) $\fR$.
In the fourth step (which comprises the bulk of the algorithm in terms of runtime), we iteratively switch between further reducing the maximum outdegree and the maximum indegree of the contention digraph.
The reductions of the maximum outdegree are obtained by a ruling set computation similar to the one in the second step; the reductions of the indegree are obtained by a variation of the approach for the third step.
The crucial point of this iterative procedure is that each reduction of the maximum indegree ensures that the subsequent ruling set computation (for reducing the maximum outdegree) can be performed efficiently on a virtual graph representing relations of \emph{larger distance} between orchids than in the previous iteration, while each reduction of the maximum outdegree ensures that the subsequent indegree-reduction step can be performed with a similarly improved parameter that guarantees that the maximum indegree indeed shrinks further.
When the distance of the relations represented by the aforementioned virtual graph exceeds $k$, the contention digraph has finally lost all of its edges, guaranteeing the property stated in the first bullet of~\cref{thm:rulingsubgraphs}.
By carefully balancing the parameters in the above approach, we ensure that for each member of $\fR$ that is removed, the distance to the closest remaining member of $\fR$ does not increase too much, thereby guaranteeing the property stated in the second bullet of~\cref{thm:rulingsubgraphs}.

\begin{algorithm}[ht!]
\caption{\RulingSubgraphs, for a family $\fR$ of connected subgraphs of at most $k$ nodes each}
\label{alg:rulingsubgraphs}
\begin{algorithmic}[1]
\For{each subgraph $H_i \in \fR$ (in parallel)}
    \State Put a token $T_i$ on the root $v_i$ of the $(t,d)$-orchid $F(H_i)$ of $H_i$, labeled with the ID of $v_i$, the orchid $F(H_i)$, and the subgraph $H_i$.
    \State Compute a walk $\sigma_i$ of length at most $2k$ for token $T_i$ on $H_i$ that visits each node of $H_i$ and returns to $v_i$ (e.g., using a DFS traversal of $H_i$).
    \State Similarly, compute a walk $\rho_i$ of length at most $2t$ for token $T_i$ on the stem of $F(H_i)$ that visits each node of the stem of $F(H_i)$ and returns to $v_i$. 
\EndFor
\State {$\EliminationWalk$} (\cref{sec:break-overlap}) \Comment{computes $\fR_1 \subseteq \fR$}

\State {$\RulingOrchids$} (\cref{sec:iniOutdeg}) \Comment{computes $\fR_2 \subseteq \fR_1$}

\If{$2 \cdot k/d < 1$}\Comment{implying an empty contention digraph by \cref{lem:rulingorchids}}
    \State Return $\fR_2$.
\Else
\State \DirectionalEliminationWalk\ (\cref{sec:direlimwalk}) \Comment{computes $\fR_3 \subseteq \fR_2$}
\For{$i = 1, \ldots, i_{\max} \leq O(\log^* k - \log^*\log^*n)$}\Comment{Until $\log^{(i)} k \in O(\sqrt[1/4]{\log^* n})$}
    \State \OutdegreeReduction\ (\cref{sec:outdegree-reduction}) \Comment{computes $\fR^a_i \subseteq \fR^b_{i-1}$}
    \State \PausingEliminationWalk\ (\cref{sec:indegree-reduction}) \Comment{computes $\fR^b_i \subseteq \fR^a_i$}
\EndFor
\Statey Consider $\fR^b_{i_{\max}}$, the result of the last iteration of \PausingEliminationWalk.
\State Compute an MIS of the contention digraph of $\fR^b_{i_{\max}}$ \Comment{computes $\fR_{\mathsf{final}} \subseteq \fR^b_{i_{\max}}$}
\State Return $\fR_{\mathsf{final}}$, the MIS computed on the previous line.
\EndIf
\end{algorithmic}
\end{algorithm}

\subsection{Step 1: Degree Reduction in the Distance-\texorpdfstring{$d$}{d} Conflict Graph of Orchids}
\label{sec:break-overlap}

Our first procedure eliminates subgraphs from the initial family $\fR$ so that each subgraph in the remaining family $\fR_1$ has its orchid overlap with the orchids of at most $\poly(t,\Delta)$ other subgraphs.
The procedure (\EliminationWalk, \cref{alg:elimwalk}) just has each subgraph walk a token on the stem of its orchid, and tokens (and their corresponding subgraphs) are deleted according to a simple rule: whenever a node contains more than one token during a round, all but one of them is eliminated.

\begin{algorithm}[htb!]
\caption{\EliminationWalk, on subgraph family $\fR$}
\label{alg:elimwalk}
\begin{algorithmic}[1]
\For{$j = 1, \ldots, 2t$}
    \State At each node containing multiple tokens, delete all but the one of highest ID.
\State Each surviving token $T_i$ takes a step of the walk $\rho_i$ on the stem of the corresponding orchid $F(H_i)$, if a step remains.
\EndFor
    \State Remove from $\fR$ the subgraphs whose token was eliminated, call the new set of subgraphs $\fR_1$.
\end{algorithmic}
\end{algorithm}

\begin{lemma}
    \label{lem:elimwalk}
    Let $\fF$ denote the set of the $(t,d)$-orchids of the subgraphs $H_i \in \fR$, and let $\fF_1 \subseteq \fF$ denote the subset of those $(t,d)$-orchids $F(H_i)$ whose corresponding token $T_i$ survived at the end of \EliminationWalk.
Then $\fF_1$ satisfies the following properties.
    \begin{itemize}
        \item Every orchid in $\fF_1$ is adjacent to at most $2t^2 \Delta^{3d}$ orchids in $\fF_1$ in the distance-$d$ conflict graph $C_{\fF_1}$, i.e., $\Delta(C_{\fF_1}) \leq 2t^2 \Delta^{3d}$.
        \item Every orchid in $\fF$ whose corresponding token was eliminated during \EliminationWalk\ (i.e., any orchid in $\fF \setminus \fF_1$) is within distance $2t$ from some orchid $F(H) \in \fF_1$ (in the input graph $G$).
    \end{itemize}
    Moreover, \EliminationWalk\ can be performed in $O(t)$ rounds.
\end{lemma}

\begin{proof}
    First, we claim that the runtime is trivially $O(t)$. Indeed, we have $2t$ iterations of a loop, each of which takes $O(1)$ rounds. Deleting a token is a $0$-round operation (it simply means that the token will not be sent to a neighboring node at the end of the loop iteration), and sending a token to the next node in its walk only requires one round to inform said next node that it receives a token, together with the token's information.
    Removing a subgraph from $\fR$ is also a $0$-round operation as this is just bookkeeping performed at the root of the orchid of the respective subgraph.

    Let us now bound the maximum degree of the distance-$d$ conflict graph of the computed orchid family $\fF_1$, i.e., the maximum number of other orchids from $\fF_1$ that can be within distance $d$ of any given orchid from $\fF_1$ (in $G$).
    Consider an orchid $F(H) \in \fF_1$.
    The stem of $F(H)$ contains at most $t$ nodes (by definition), and for any of these $t$ nodes, there are at most $\Delta^{3d}$ nodes at distance $3d$. Hence, by the definition of a $(t,d)$-orchid, in total there are at most $t \Delta^{3d}$ nodes at distance $2d$ or less from $F(H)$.

    Let $F(H_1), \dots, F(H_z)$ be the orchids in $\fF_1$ at distance $d$ or less from $F(H)$ in $G$.
    For each such orchid $F(H_i)$, let $w_i$ be the node of minimum ID amongst all nodes that are both contained in the stem of $F(H_i)$ and at distance at most $2d$ from $F(H)$. (Such a node $w_i$ must exist for each $1 \leq i \leq z$ by the definition of a $(t,d)$-orchid). Let $t_i \leq 2t$ be the iteration in which node $w_i$ is first visited in the walk $\rho_i$. At most $t \Delta^{3d}$ vertices are possible candidates for $w_i$, and $2t$ values for $t_i$, for a total of at most $2t^2\Delta^{3d}$ distinct candidate pairs $(w_i,t_i)$.

    Now, suppose that for two distinct indices $i, j \in \{ 1, \dots, z\}$, we have $(w_i,t_i) = (w_j,t_j)$.
    This would imply that in iteration $t_i=t_j$ of the loop, the tokens $T_i, T_j$ corresponding to the two orchids $F(H_i)$ and $F(H_j)$ were on the same node $w_i = w_j$. But if that were true, then one of the two tokens would have been removed during \EliminationWalk, making it impossible for the corresponding orchid to be contained in $\fF_1$. Therefore, $\fF_1$ cannot contain more orchids at distance $2d$ or less from $F(H)$ than there are possible distinct pairs $(w_i,t_i)$, i.e., at most $2t^2 \Delta^{3d}$.
    This implies that every orchid in $\fF_1$ is adjacent to at most $2t^2 \Delta^{3d}$ orchids in $\fF_1$.

    Finally, consider an orchid $F(H_0)\in \fF$ whose corresponding token $T_0$ was eliminated during the process, and let $u_0$ be the node at which $T_0$ starts its walk. Let $u_1$ be the node at which $T_0$ was eliminated by the token $T_1$ corresponding to another orchid $F(H_1)\in \fF$. Similarly, if it exists, for $i\geq 1$, let $F(H_{i+1})$ be the orchid whose corresponding token $T_{i + 1}$ eliminates the token $T_i$ corresponding to graph $F(H_i)$ during \EliminationWalk, and $u_{i+1}$ the node at which the elimination occurs. Let $\imax$ be the last index of this process, i.e., the value $\imax$ such that token $T_{\imax}$ is not eliminated by another token during \EliminationWalk. Let $u_{\imax+1}$ be the last node visited by $T_{\imax}$ during \EliminationWalk. For each $i \in \set{0,\ldots, \imax}$, consider the walk $P_i$ taken by token $T_{i}$ from $u_{i}$ to $u_{i+1}$.
    Let the \emph{postmortem walk of token $T_0$} be defined as the concatenation $P_0,\ldots,P_{\imax}$.
    The postmortem walk of $T_0$ is a walk of length $2t$, as it tracks the movement of a token and the tokens that recursively eliminate it (and each other) during the $2t$-size walk of \EliminationWalk. As a result, the root of the orchid $F(H_0)$ is within distance $2t$ from the root of $F(H_{\imax})$, which is contained in $\fF_1$.
\end{proof}

\subsection{Step 2: Initial Outdegree Reduction}
\label{sec:iniOutdeg}

The degree reduction from the previous step allows us to further break symmetry between orchids at a more reasonable cost than if it had not been done. After \RulingOrchids, subgraphs in the new family $\fR_2$ have a distance of at least $d+1$ between their orchids.

\begin{algorithm}[ht!]
\caption{\RulingOrchids, on subgraph family $\fR_1$}
\label{alg:rulingorchids}
\begin{algorithmic}[1]
\Statex Let $\fF_1$ denote the set of orchids of the members of $\fR_1$ (which survived \EliminationWalk). 
\State Compute a $(2,\log (\Delta (C_{\fF_1})))$-ruling set on the distance-$d$ conflict graph $C_{\fF_1}$ of $\fF_1$. Let $\fF_2$ be the set of orchids selected into the ruling set.
\State Remove from $\fR_1$ those members $H \in \fR_1$ for which $F(H) \in \fF_1 \setminus \fF_2$, and call the obtained set of subgraphs $\fR_2$.
\end{algorithmic}
\end{algorithm}
\begin{lemma}
    \label{lem:rulingorchids}
For the sets $\fF_2$ and $\fR_2$ computed by \RulingOrchids, it holds that
    \begin{itemize}
        \item the distance between any two orchids in $\fF_2$ is at least $d + 1$ in $G$,
        \item for every orchid $H \in \fF_1 \setminus \fF_2$, there is an orchid $H' \in \fF_2$ such that the distance between $H$ and $H'$ is in $O((t+d)(d\log \Delta + \log t))$, and 
\item the contention digraph $C(\fR_2, G)$ has maximum outdegree $2k/d$.\end{itemize}
    Moreover, \RulingOrchids\ can be performed in $O((t+d) (d \log \Delta + \log t + \log^* n))$ rounds.
\end{lemma}

\begin{proof}
Consider the distance-$d$ conflict graph $C_{\fF_1}$ of $\fF_1$, i.e., the graph whose nodes are the members of $\fF_1$ and for which there is an edge between two orchids if they are at distance $d$ or less in $G$. By \cref{lem:linial,lem:ruling-set}, and the fact that $\Delta(C_{\fF_1})^{1/\log(\Delta(C_{\fF_1}))}$ is a constant, a $(2,\log(\Delta(C_{\fF_1})))$-ruling set on $C_{\fF_1}$ can be computed in $O(\log(\Delta(C_{\fF_1})) + \log^* n)$ communication rounds, assuming that $C_{\fF_1}$ is the input graph on which the communication in the LOCAL model takes place. Since, for each orchid in $\fF_1$, the maximum ID node on its stem is within distance $2t+3d$ (in $G$) of the maximum ID node on the stem of each of its neighboring orchids in $C_{\fF_1}$, the ruling set computation, and therefore the entire execution of \RulingOrchids, can be performed in $O((t+d)(\log(\Delta(C_{\fF_1})) +\log^* n))$ rounds with $G$ as communication graph.
By \cref{lem:elimwalk}, we have $\Delta(C_{\fF_1}) \leq 2t^2\Delta^{3d}$, and we obtain that \RulingOrchids\ can be performed in $O((t+d)(d \log \Delta + \log t +\log^* n))$ rounds.

Moreover, the ruling set computed on the conflict graph $C_{\fF_1}$ is a $(2, O(d\log \Delta + \log t))$-ruling set. Accounting again for the dilation of the virtual graph $C_{\fF_1}$ compared to the underlying graph $G$, we obtain that every two orchids $H, H' \in \fF_1$ are at distance at least $d + 1$ from each other, and that every subgraph $H \in \fF_2 \setminus \fF_1$
has a member of $\fF_1$ at distance at most $O((t+d)(d \log \Delta + \log t))$ in $G$ (where we use that every $(t,d)$-orchid has diameter at most $t+2d$, by definition).

Finally, observe that the fact that every member of $\fR_2 \subseteq \fR$ is a connected subgraph of at most $k$ nodes, together with the fact that any two distinct members of $\fF_2$ have distance at least $d+1$ in $G$, implies that for each $H \in \fR_2$, the number of orchids in $\fF_2$ different from $F(H)$ that (nontrivially) intersect $H$ is bounded by $2k/d$.
Hence, by the definition of a contention digraph, we obtain the same bound on the outdegree of the contention digraph $C(\fR_2, G)$.

\end{proof}

\subsection{Step 3: Initial Indegree Reduction} 
\label{sec:direlimwalk}

\begin{algorithm}[htb!] 
\caption{\DirectionalEliminationWalk, on subgraph family $\fR_2$}
\label{alg:direlimwalk} 
\begin{algorithmic}[1]

\Statex Throughout the algorithm, whenever there are multiple tokens at the same node at the same time, all but the one of highest ID are deleted.

\smallskip

\Statex{Each non-deleted token $T_i$ (i.e., each token corresponding to some $H_i \in \fR_2$) performs the steps described in the following loop}
\Repeat 
\State $T_i$ takes a step of the walk $\sigma_i$ on its subgraph $H_i$. \Comment{$\sigma_i$ contains at most $2k$ steps}
    \If{$T_i$ enters the orchid $F(H)$ of a subgraph $H \neq H_i$ in $\fR_2$ for the first time}
        \State Consider a shortest path to the encountered orchid's stem. \Comment{precomputed}
        \State $T_i$ moves on this path, until it reaches the stem of the orchid. \Comment{at most $d$ steps}
        \State $T_i$ walks back to the starting node of the path. \Comment{at most $d$ steps}
    \EndIf
\Until{$T_i$ did all steps of its walk $\sigma_i$ on its subgraph $H_i$, or was deleted}
\State Return $\fR_3$, the subgraphs members of $\fR_2$ whose token was not deleted.
\end{algorithmic} 
\end{algorithm}

Note that in \DirectionalEliminationWalk, at the same overall time step different tokens might be in different iterations of the for loop.
Note further that if $q \cdot k/d < 1$ for the universal constant $q$ used in~\cref{alg:rulingsubgraphs} (\RulingSubgraphs), then \RulingSubgraphs\ terminates after executing subroutine \RulingOrchids\ and outputs $\fR_2$.
Hence, assume throughout \cref{sec:direlimwalk,sec:alternating-reduction} that $d \in O(k)$.
 
\begin{lemma} 
    \label{lem:dirElimWalk}
    For the set $\fR_3$ computed by \DirectionalEliminationWalk\, the maximum indegree of the contention digraph $C(\fR_3,G)$ is at most $6 k \cdot t$, and every subgraph in $\fR_2$ is within distance $6k$ of some subgraph in $\fR_3$.
    Moreover, \DirectionalEliminationWalk\ can be performed in $O(k)$ rounds.
\end{lemma} 
\begin{proof}
     We start by observing that the lengths of the walks taken by the tokens in \DirectionalEliminationWalk\ are bounded by $6k$: this simply follows from the facts that each token enters at most $2k/d$ different orchids for the first time (due to \cref{lem:rulingorchids}), that each walk towards the stem of an orchid and back is of length $2d$, by the definition of a $(t,d)$-orchid, and that each token walks at most $2k$ steps on its subgraph.
     This implies that \DirectionalEliminationWalk\ can be performed in $O(k)$ rounds.
     (Note that the computation of a shortest path towards an orchid's stem can be performed in $0$ rounds during the algorithm by simply storing the entire topology of any orchid $H$ in every node of $H$, which can be achieved by a precomputation of $O(t+d) \subseteq O(k)$ rounds, by the definition of a $(t,d)$-orchid. Here, we make use of the facts that $t \leq k$ and $d \in O(k)$.)
        
     Now consider some subgraph $H_i \in \fR_3$, i.e., a subgraph for which the corresponding token $T_i$ was not eliminated during \DirectionalEliminationWalk.
     For each in-neighbor $H_s$ of $H_i$ in $C(\fR_3,G)$, the corresponding token $T_s$ must have had its token run its full walk without being deleted. At one of the $O(k)$ steps of the walk, $T_s$ was located in some node of the stem of the orchid $F(H_i)$ of $H_i$. For each of the overall $6k$ time steps, there can be at most $t$ tokens that are located at some node of the stem of $F(H_i)$ at this time step and do not get deleted.
     Hence, we obtain that there are at most $6k \cdot t$ in-neighbors of $H_i$, which provides the desired bound on the maximum indegree of $C(\fR_3,G)$.

     Next, let $H'$ be a subgraph from $\fR_2$, and $T'$ the corresponding token.
     Consider the postmortem walk of token $T'$ (as defined in the proof of~\cref{lem:elimwalk}).
     Analogously to the proof of~\cref{lem:elimwalk}, we obtain that this postmortem walk is of length $6k$ and ends in a node of a subgraph whose corresponding token survived \DirectionalEliminationWalk. As a result, each subgraph from $\fR_2$ is within distance $6k$ of some subgraph from $\fR_3$.
\end{proof}

\subsection{Step 4: Alternating (In/Out)degree Reduction}
\label{sec:alternating-reduction}

\subsubsection{Step 4a: Outdegree reduction by ruling set}
\label{sec:outdegree-reduction}

\begin{definition}[Sibling graph]
    \label{def:sibling-graph}
    Let $\fR$ be a family of subgraphs of the input graph $G$. For any integer $x \geq 0$, the \emph{$x$-sibling graph} $S(\fR,G,x)$ of $\fR$ is the virtual graph defined via
    \begin{itemize}
        \item $V(S(\fR,G,x)) = \fR$,
        \item $E(S(\fR,G,x)) = \bigcup_{H \in \fR} \set{ \{H',H''\} \mid H', H'' \in \fR, H' \neq H'', \text{ and there exists a pair } (v',v'') \in (V(H)\cap V(F(H')))\times (V(H)\cap V(F(H''))) \text{ such that } \dist_H(v',v'') \leq x}$.
    \end{itemize}
    In other words, two subgraphs $H'$ and $H''$ are connected in $S(\fR,G,x)$ whenever their orchids intersect a third subgraph $H$ and the intersections points are at distance at most $x$ in $H$. Note that $H$ is not necessarily different from $H'$ (or $H''$): we also add an edge between $H$ and $H'$ if the orchid of $H'$ intersects $H$ close to the orchid of $H$. We call $x$ the \emph{connecting distance} of the sibling graph.
\end{definition}

Note that if $H \not \in \set{H',H''}$ in the above definition, then $H'$ and $H''$ are out-neighbors of $H$ in the contention digraph $C(\fR,G)$, and, conversely, $H$ is an in-neighbor of $H'$ and $H''$. If $H'=H$ and the edge $\{H',H''\}$ exists in the sibling graph, then $H'=H$ is an in-neighbor of $H''$.
This connection with the contention digraph implies the following bound on the maximum degree of the sibling graph, where we use $\indeg_{\max}(\fR)$ and $\outdeg_{\max}(\fR)$ to denote the maximum in- and outdegree, respectively, of the contention digraph $C(\fR,G)$..

\begin{proposition}
\label{prop:sibling-graph-degree}
    Let $\fR$ be a family of subgraphs of $G$ and $x \geq 0$ any integer.
Then the sibling graph 
    $S(\fR,G,x)$
    has maximum degree $(\indeg_{\max}(\fR)+1) \cdot \outdeg_{\max}(\fR)$.
\end{proposition}

Note that the bound given by \cref{prop:sibling-graph-degree} is independent of the connecting distance $x$.

We now provide~\cref{alg:outdegree-reduction}, which is the first of the two subroutines (parameterized by some iteration counter $i$) that are performed iteratively and alternatingly during \cref{alg:rulingsubgraphs}.
Note that in each iteration $i \geq 2$, \cref{alg:outdegree-reduction} depends on the output $\fR^b_{i-1}$ of the second subroutine (\cref{alg:pausingelimwalk}) in iteration $i - 1$.

\begin{algorithm}[htb!]
\caption{\OutdegreeReduction\ in iteration $i \geq 1$, on subgraph family $\fR^b_{i-1}$}
\label{alg:outdegree-reduction}
\begin{algorithmic}[1]
\Statex Let $\dlogsup{i} = 4k/(\log^{(i)} k)^2$ and set $\fR^b_{0} := \fR_3$.
\State Compute a $(2, \log \Delta(S(\fR^b_{i-1},G,\dlogsup{i})))$-ruling set $\fR^a_i \subseteq \fR^b_{i-1}$ on $S(\fR^b_{i-1},G,\dlogsup{i})$.
\end{algorithmic}
\end{algorithm}

\begin{lemma}
    \label{lem:outDegreeReduction}
    Consider iteration $i$ of~\cref{alg:outdegree-reduction}.
    Suppose that for the maximum indegree $\indeg_{\max}(\fR^b_{i-1})$ and the maximum outdegree $\outdeg_{\max}(\fR^b_{i-1})$ of the contention digraph $C(\fR^b_{i-1},G)$ it holds that $\indeg_{\max}(\fR^b_{i-1}) \in O((\log^{(i-1)} k)^4)$ and $\outdeg_{\max}(\fR^b_{i-1}) \in O((\log^{(i-1)} k)^2)$, where we set $\log^{(0)} k := k$. Then for the set $\fR^a_i$ returned by \cref{alg:outdegree-reduction}, it holds that
    \begin{itemize}
        \item the maximum outdegree of the contention digraph $C(\fR^a_i, G)$ is at most $(\log^{(i)} k)^2/2$,
        \item the maximum indegree of the contention digraph $C(\fR^a_i, G)$ is in $O((\log^{(i-1)} k)^4)$, and
        \item every subgraph contained in $\fR^b_{i-1} \setminus \fR^a_i$ is at distance at most $O(k / \log^{(i)} k + (t+d)\log^{(i)} k)$ from some subgraph contained in $\fR^a_i$.
    \end{itemize}
    Moreover, \cref{alg:outdegree-reduction}
can be performed in $O((t+d+\dlogsup{i})(\log^{(i)} k + \log^*n))$ rounds.
\end{lemma}

\begin{proof}
    Each node in the $\dlogsup{i}$-sibling graph $S(\fR^b_{i-1},G,\dlogsup{i})$ is a subgraph from the family $\fR^b_{i-1}$.
    We can simulate an algorithm executed on the sibling graph on our input graph $G$ by letting each subgraph $H \in \fR^b_{i-1}$ be simulated by the root $r_H$ of the orchid $F(H)$ of $H$.
    For any two subgraphs connected by an edge in the sibling graph, the roots of their orchids are at distance at most $\dlogsup{i} + 2d + 2t$ in $G$, which implies that a round of communication over the sibling graph can be simulated in $\dlogsup{i} + 2d + 2t$ rounds of communication over $G$.
    
    Let $\Delta' := \Delta(S(\fR^b_{i-1},G,\dlogsup{i}))$ be the maximum degree of the sibling graph. By \cref{lem:linial,lem:ruling-set}, a $(2,O(\log(\Delta')))$-ruling set on the sibling graph can be computed in $O(\log \Delta' + \log^* n)$ rounds of communication over the sibling graph.
By the hypotheses of the lemma and \cref{prop:sibling-graph-degree}, $\Delta' \in O(( \log^{(i-1)} k)^6)$, and thus $\log \Delta' \in O(\log^{(i)} k)$.

    Put together, the cost of simulating in $G$ a round of communication over the sibling graph and the runtime of the ruling set algorithm over the sibling graph result in a complexity of at most $O((t+d+\dlogsup{i})(\log^{(i)} k + \log^*n))$ rounds.
    The maximum distance of $O(\log \Delta') \subseteq O(\log^{(i)} k)$ between any node of the sibling graph and a node that joined the ruling set implies that for each subgraph in $\fR^b_{i-1}\setminus \fR^a_i$ the distance to the closest subgraph in $\fR^a_i$ is in $O((t+d+\dlogsup{i})\log^{(i)} k) = O(k / \log^{(i)} k + (t+d)\log^{(i)} k)$.

    Next, consider two out-neighbors $H',H'' \in \fR^a_i$ of a subgraph $H\in \fR^a_i$ in the contention digraph $C(\fR^a_i,G)$.
    Since $H',H'' \in \fR^a_i$, they were not adjacent in the sibling graph of the family $\fR^b_{i-1}$. Therefore, the intersections of their orchids with $H$ are at distance at least $\dlogsup{i}+1$ in $H$.
    This implies that for any two nodes $v' \in V(F(H')) \cap V(H)$ and $v'' \in V(F(H'')) \cap V(H)$, the distance-$\dlogsup{i}/2$ neighborhoods of $v'$ and $v''$ in $H$ contain at least $\dlogsup{i}/2$ nodes each (by the connectivity of $H$), and that these two neighborhoods do not intersect.
Therefore, the outdegree of $H$ in $C(\fR^a_i)$ is at most $2\card{H} / \dlogsup{i} \leq (\log^{(i)} k)^2/2$.

    Finally, we observe that the maximum indegree of $C(\fR^a_i,G)$ is in $O((\log^{(i-1)} k)^4)$ since $\fR^a_i \subseteq \fR^b_{i-1}$ and the lemma assumes $\indeg_{\max}(\fR^b_{i-1}) \in O((\log^{(i-1)} k)^4)$.

\end{proof}

\subsubsection{Step 4b: Indegree reduction by pausing walk}
\label{sec:indegree-reduction}

Similarly to how {\DirectionalEliminationWalk} from \cref{sec:direlimwalk} is able to reduce the maximum indegree of the considered contention digraph to something of the order of the maximum outdegree of the contention digraph,
in this section, we describe and analyze an algorithm for reducing the indegree of the contention digraph following the outdegree reduction from the previous section.
Before describing the algorithm, we introduce the notion of a set of branch nodes.

\begin{definition}[Branch nodes]
\label{def:branch-nodes}
    A set $B$ of $x$ \emph{branch nodes} of a subgraph $H$ with $(t,d)$-orchid $F(H)$ with stem $F'(H)$ is a set with the following properties:
    \begin{itemize}
        \item $\card{B} = x$,
        \item $B \subseteq V(F'(H))$, i.e., $B$ is fully contained in the node set of the stem of $F(H)$,
        \item $\forall v \in V(F'(H))$, $\dist_{F'(H)}(v,B) \leq t/x$.
    \end{itemize}
    Intuitively, $B$ is a distance-$(t/x)$ dominating set of the stem of $F(H)$ of size $x$. 
\end{definition}

\begin{proposition}
    For each $x \geq 1$, a set of $x$ branch nodes as described in \cref{def:branch-nodes} exist.
\end{proposition}
\begin{proof}
    Consider a walk $\sigma$ of length $2t$ over the stem $F'(H)$ of the orchid of $H$ that visits all node of $F'(H)$, and index the visited nodes with indices from $0$ to $2t$ (where nodes may receive multiple indices). Take as branch nodes the nodes of indices $\min(2t,\ceil{t/x} + j \cdot \ceil{2t/x})$ for $j \in \set{0,\ldots,x-1}$. This defines a set $B$ of size $x$ contained in the node set of the stem s.t.\ every node of the stem is at distance at most $t/x$ from the nearest node in $B$.
    (Note that the last property also holds for the node with index $0$ as this node is identical to a node with non-zero index.)
\end{proof}

We now describe {\PausingEliminationWalk}, our algorithm for reducing the indegree of the contention digraph. At a high level, {\PausingEliminationWalk} is a generalization of {\DirectionalEliminationWalk} where tokens encountering an orchid for the first time walk to the stem of the orchid and \emph{pause} there for some number of rounds $\ell=4k/p$, i.e., do not move for the next $\ell$ rounds. While in {\DirectionalEliminationWalk}, two tokens reaching the stem of the same orchid at different points in time do not meet each other there, the addition of a pause allows tokens to meet even if they reach the same stem at different points in time, provided those points in time are not too distant.
Furthermore, {\PausingEliminationWalk} makes use of a subset of the nodes of the stem, the branch nodes, to gather the tokens on a smaller number of nodes in the stem. This increases even further the opportunities that tokens have to eliminate each other.
{\DirectionalEliminationWalk} basically corresponds to {\PausingEliminationWalk} with pauses of length $0$ with the whole stem as branch nodes.

Each token of a subgraph $H$ that runs the full course of its walk over $H$ is stopped by each orchid overlapping with $H$---each of which contributes to the outdegree of $H$ in the contention digraph. When entering a previously unseen orchid, the token adjusts its walk by inserting a walk to and from the nearest branch node of the orchid.
The length of pauses and number of branch nodes are chosen respectively low and high enough so that every token can complete its entire walk in $O(k)$ rounds, even if it encounters as many as $\outdeg_{\max}(\fR^a_i)$ orchids during its walk, and is at the maximum distance of $d+t/x$ to the nearest branch node when encountering each orchid.

\begin{algorithm}[htb!]
\caption{\PausingEliminationWalk\ in iteration $i \geq 1$, on subgraph family $\fR^a_{i}$}
\label{alg:pausingelimwalk}
\begin{algorithmic}[1]
\Statex Set $p := 2\outdeg_{\max}(\fR^a_i)$ and $\ell := \lceil4k/p\rceil$.
\State Compute a set of $p$ branch nodes of each subgraph in $\fR^a_i$.
\State Each non-deleted token $T_s$ (i.e., each token corresponding to some $H_s \in \fR^a_i$) performs the steps described in the following loop. Whenever there are multiple tokens at the same node at the same time, all but the one of the highest ID are deleted.
\For{$j = 1, \ldots, 2k$}
    \State $T_i$ takes a step of the walk $\sigma_i$ on its subgraph $H_i$, if a step remains.
    \State If token $T_i$ enters an orchid of a subgraph $H \neq H_i$ in $\fR^a_i$ for the first time during $\sigma_i$, it walks to the nearest branch node on the stem of that orchid, makes a pause of length $\ell$ there (i.e., does not move for $\ell$ steps), and then walks back to the node in $H_i$ where it encountered the orchid. \EndFor
\State Remove from $\fR^a_i$ all members whose corresponding token was deleted and call the resulting set of subgraphs $\fR^b_i$.
\end{algorithmic}
\end{algorithm}

\begin{lemma}
\label{lem:pausingwalk}
    Consider iteration $i$ of \cref{alg:pausingelimwalk}.
    Suppose that for the maximum indegree $\indeg_{\max}(\fR^a_i)$ and the maximum outdegree $\outdeg_{\max}(\fR^a_i)$ of the contention digraph $C(\fR^a_i,G)$ it holds that $\indeg_{\max}(\fR^a_i) \in O((\log^{(i-1)} k)^4)$ and $\outdeg_{\max}(\fR^a_i) \leq (\log^{(i)} k)^2/2$.
Then, for the set $\fR^b_i$ returned by~\cref{alg:pausingelimwalk}, it holds that
\begin{itemize}
        \item the maximum outdegree of the contention digraph $C(\fR^b_i,G)$ is at most $(\log^{(i)} k)^2/2$,
        \item the maximum indegree of the contention digraph $C(\fR^b_i,G)$ is in $O((\log^{(i)} k)^4)$, and
\item every subgraph contained in $\fR^a_i \setminus \fR^b_i$ is at distance at most $O(k)$ from some subgraph contained in $\fR^b_i$.
\end{itemize}
    Moreover, \cref{alg:pausingelimwalk} can be performed in $O(k)$ rounds.
\end{lemma}
\begin{proof}
    Let us analyze how many rounds each token spends either paused, walking on its own subgraph, and walking between its subgraph and a branch node whose orchid it encountered.
    
    As every subgraph in $\fR^a_i$ has a maximum outdegree of at most $p/2$ in the contention digraph $C(\fR^a_i,G)$, each token waits only once on the stem of each encountered orchid, and each token pauses for $\ell$ rounds once at a branch node where it pauses, the total waiting time for each token is at most $p/2 \cdot \ell \leq p/2 \cdot (4k/p + 1) \leq 3k$ (as $p/2 = \outdeg_{\max}(\fR^a_i) \leq k$). 
    
    On its own subgraph $H$, each token does a walk of length at most $2k$. All that remains to analyze are the walks between a token's subgraph and the nearest branch node of an encountered orchid.

Walking to the nearest branch node is a matter of at most $d+t/p$ steps, since the stem can be reached in $d$ steps from any point of the orchid, and every node in the stem is at distance at most $t/p$ from the nearest branch node. Since each token encounters at most $p/2$ orchids, and every walk to the nearest branch node is done in both directions, at most $p\cdot d+t$ steps are spent doing these parts of the walk. The stem being a subgraph of $H$ implies $t \leq k$, and we have $p \in O(k/d)$
from the very first outdegree reduction (\RulingOrchids). Thus at most $O(k)$ steps of this kind are taken by each token, and we conclude that the overall number of time steps for a token to finish its walk is bounded by $O(k)$ as well.
    Note that the computation of the branch nodes at the beginning of the algorithm can be performed in $0$ rounds (allowing some precomputation whose complexity is dominated by $O(k)$) analogously to the proof of~\cref{lem:dirElimWalk}.
    Hence, \cref{alg:pausingelimwalk} can be performed in $O(k)$ rounds.

    With the same argument as in the proofs of \cref{lem:elimwalk,lem:dirElimWalk} (based on the notion of a postmortem walk), we also obtain that each subgraph contained in $\fR^a_i \setminus \fR^b_i$ is at distance at most $O(k)$ from some subgraph contained in $\fR^b_i$.

    Next, we bound the maximum indegree of $C(\fR^b_i)$.
    Let $H' \in \fR^b_i$ be an in-neighbor of $H$ in $C(\fR^b_i,G)$, which implies that $H'$ intersects the orchid of $H$. The token $T'$ corresponding to $H'$ therefore visited a branch node of $H$ during its walk, and stayed there for $\ell = \lceil 4k/p \rceil$ rounds. Any other subgraph $H'' \in \fR^b_i$ that is an in-neighbor of $H$ cannot have had its token $T''$ visit the same branch node during any of these $\ell$ rounds, for otherwise the tokens of $H'$ and $H''$ would have met and one of them would have been eliminated. Therefore, each branch node is visited by at most $O\parens*{\frac{k}{\ell}}= O(p)$ surviving tokens during the walk. Since each orchid contains at most $p$ branch nodes, we obtain the upper bound of $O(p^2)$ on the maximum indegree of $C(\fR^b_i,G)$.
    By observing that $p = 2\outdeg_{\max}(\fR^a_i) \leq (\log^{(i)} k)^2$, we obtain the bound on the maximum indegree of $C(\fR^b_i,G)$ stated in the lemma.

    Finally, we observe that the maximum outdegree of $C(\fR^b_i,G)$ is at most $(\log^{(i)} k)^2/2$ since $\fR^b_i \subseteq \fR^a_{i}$ and the lemma assumes $\outdeg_{\max}(\fR^a_{i}) \leq (\log^{(i)} k)^2/2$.

\end{proof}

\subsection{Putting Everything Together}

In this section, we finally prove~\cref{thm:rulingsubgraphs} by showing that the set of subgraphs computed by~\cref{alg:rulingsubgraphs} satisfies the conditions given in~\cref{thm:rulingsubgraphs} and bounding the runtime of~\cref{alg:rulingsubgraphs}.

\RulingSubgraphsThm*

\begin{proof}
We analyze the effect of each part of the algorithm, most notably its contribution to the algorithm's runtime and to the maximum distance of every deleted subgraph of the input family $\fR$ to the nearest subgraph in the output family $\fR'$.

    \paragraph{Initialization (Steps 1 to 4).}
    We start by observing that every subgraph $H_i$ in $\fR$ can discover its entire topology in $O(k)$ rounds, and in particular, discover its root node $v_i$ so as to put a token on it.
    This allows~\cref{alg:rulingsubgraphs} to compute a walk $\sigma_i$ of length $< 2|H_i|=2k$ starting from $v_i$ and visiting every node of $H_i$. Such a walk exists and is easily computed, since, for example, a tree traversal of any spanning tree of $H_i$ rooted at $v_i$ has the stated length of $< 2\card{H_i}$.

    \paragraph{Preventing Orchid Overlap (Steps 5 and 6).}
    \cref{alg:rulingsubgraphs} then applies {\EliminationWalk} and {\RulingOrchids}, which by \cref{lem:elimwalk,lem:rulingorchids} results in orchids of the the resulting family $\fR_2$ being at distance at least $d+1$ in $G$ and bounds the maximum outdegree of the contention digraph of $\fR_2$ to $O(k/d)$.
    This part of the algorithm takes $O((t+d)(d\log \Delta + \log t + \log^*n))$ rounds, and
    the distance between every deleted subgraph $H \in \fR \setminus \fR_2$ and the output family $\fR_2$ is at most $O((t+d)(d\log \Delta + \log t))$.

    The test in step 7 implies that from there onwards, $k \geq 2d$.
    
    \paragraph{Initial Indegree Reduction (Step 10).}
    {\DirectionalEliminationWalk} then produces a family $\fR_3$ whose contention digraph has maximum indegree $O(k^2)$, by \cref{lem:dirElimWalk}. \Cref{lem:dirElimWalk} additionally guarantees that this operation is performed in $O(k)$ rounds, and every removed subgraph (in $\fR_2 \setminus \fR_3$) is at a distance of at most $O(k)$ of a subgraph in $\fR_3$. 

    \paragraph{Iterated Degree Reduction (Steps 11 to 13).}
    The algorithm then repeatedly applies {\OutdegreeReduction} and {\PausingEliminationWalk}, in alternation. In the $i$th iteration of the for loop, we start from a subgraph family $\fR^b_{i-1}$, which {\OutdegreeReduction} narrows down to a family $\fR^a_i$, which {\PausingEliminationWalk} then thins out to a family $\fR^b_i$.
    Each iteration of the loop results in a family with smaller upper bounds on the in- and outdegrees of its contention digraph. 
    More precisely, it maintains the following two properties:
    \begin{itemize}
        \item The contention digraph $C(\fR^b_{i},G)$ has maximum outdegree $O((\log^{(i)} k)^2)$ and maximum indegree $O((\log^{(i)} k)^4)$ (by \cref{lem:pausingwalk}),
        \item The contention digraph $C(\fR^a_{i},G)$ has maximum outdegree $O((\log^{(i)} k)^2)$ and maximum indegree $O((\log^{(i-1)} k)^4)$ (by \cref{lem:outDegreeReduction}).
    \end{itemize}
    
    \cref{alg:rulingsubgraphs} applies the two subroutines until the maximum in- and outdegrees in the contention digraph are at most $O(\log^* n)$. As each iteration of the loop reduces the maximum in- and outdegrees of the contention digraph to a $\poly \log$ of their previous values, it only takes $O(\log^* k)$ iterations of the loop to get to in- and outdegrees of $O(\log^* n)$.
    During those loop iterations, by \cref{lem:pausingwalk}, {\PausingEliminationWalk} contributes a total of $O(k \log^*k)$ rounds to the complexity of the algorithm and $O(k \log^*k)$ to the maximum distance of deleted subgraphs to surviving ones.
    
    We now analyze the contribution of {\OutdegreeReduction} to the complexity of the algorithm and to the maximum distance of deleted subgraphs to the surviving family. Let $i_{\max} \leq O(\log^* k - \log^* \log^* n)$ be the index of the last iteration.

    The $i$th call to {\OutdegreeReduction} contributes $O((t+d + \dlogsup{i})\cdot (\log^{(i)} k + \log^* n))$ to the complexity, by \cref{lem:outDegreeReduction}, where $\dlogsup{i} \in O(k / (\log^{(i)} k)^2)$. Summing their contributions gives an overall bound of order at most
    \begin{align*}
        &\sum_{i=1}^{i_{\max}} (t+d + \dlogsup{i})\cdot (\log^{(i)} k + \log^* n)
        \\
        \leq\ & (t+d) i_{\max} \log^* n 
        + (t+d)\sum_{i=1}^{i_{\max}} \log^{(i)} k
        + \sum_{i=1}^{i_{\max}} \dlogsup{i} \cdot (\log^{(i)} k + \log^* n)
        \\
        \leq\ &  (t+d) i_{\max} \log^* n 
        + (t+d)\sum_{i=1}^{i_{\max}} \log^{(i)} k
        + k\cdot \sum_{i=1}^{i_{\max}} \frac{1}{\log^{(i)} k} 
        + k\cdot \sum_{i=1}^{i_{\max}} \frac{\log^* n}{(\log^{(i)} k)^2}\ .
    \end{align*}

    For the first sum of this last line, we have that $\sum_{i=1}^{i_{\max}} \log^{(i)} k < 2\log k$ since the terms decrease faster than a geometric series. This gives us an $O((t+d)\log k)$ term.
    For the second sum of this last line, we have that $k\cdot \sum_{i=1}^{i_{\max}} \frac{1}{\log^{(i)} k} \leq O(k)$ because for all but $O(1)$ many of the $i_{\max} \leq O(\log^* k)$ iterations, we have $\log^{(i)} k \geq \log^* n$. 
    Indeed, recall that $\log^{(i_{\max})} k \leq \sqrt[1/4]{\log^* n}$, and $\log^{(i_{\max}-1)} k > \sqrt[1/4]{\log^* n}$. When $n$ is above a sufficiently large universal constant, this means that $\log^{(i)} k> \log^* n > \log^* k$ for all $i \leq i_{\max}-2$.
    For the third sum, we similarly obtain $k\cdot \sum_{i=1}^{i_{\max}} \frac{\log^* n}{(\log^{(i)} k)^2} \leq O(k \log^* n)$.
    In total, this gives a bound of
    \[
        O((t+d)(i_{\max}\log^*n + \log k) + k\log^*n)\ ,
    \]
    which simplifies to $O((t+d)(\log^*n + \log k) + k\log^*n)$ by remarking that $\log k$ dominates $i_{\max}\log^*n $ when $k > 2^{(\log^* n)^2}$, and $i_{\max} = O(1)$ when $k \leq 2^{(\log^* n)^2}$.

    For the contribution of {\OutdegreeReduction} to the distance of deleted subgraphs to surviving ones, the total contribution is at most of order
    \begin{align*}
        & \sum_{i=1}^{i_{\max}} ( \frac{k}{\log^{(i)} k} + (t+d)\log^{(i)} k)
        \\
        = \ & \sum_{i=1}^{i_{\max}} \frac{k}{\log^{(i)} k} + (t+d)\sum_{i=1}^{i_{\max}}\log^{(i)} k
    \end{align*}

    For the rightmost sum on the last line, note that $\sum_{i=1}^{i_{\max}} \log^{(i)} k < 2\log k$. This gives a contribution of $O((t+d)\log k)$ for the rightmost part of this last equation. For the leftmost sum, recall that $\log^{(i_{\max})} k\leq \sqrt[1/4]{\log^* n}$, and $\log^{(i_{\max}-1)} k> \sqrt[1/4]{\log^* n}$. When $n$ is above a sufficiently large universal constant, this means that $\log^{(i)} k> \log^* n > \log^* k$ for all $i \leq i_{\max}-2$. The sum can thus be decomposed as
    
    \[
    \sum_{i=1}^{i_{\max}} \frac{k}{\log^{(i)} k} \leq \sum_{i=1}^{i_{\max}-2} \frac{k}{\log^{(i)} k} + \frac{k}{\log^{(i_{\max}-1)} k} + \frac{k}{\log^{(i_{\max})} k} \leq O(k) + k + k\ .
    \]

    This gives a total increase of the distance to the nearest surviving subgraph of at most $O((t+d)\log k + k)$ for the for-loop.

    \paragraph{Final MIS Computation (Step 14).}
    After the for-loop, the contention digraph has maximum in- and outdegree $\Delta' \leq O(\log^* n)$.
    We now compute an MIS on the contention digraph, ignoring the orientation of the edges. 
    Computing this MIS on the contention digraph is done in $O(\Delta' + \log^* n)$ rounds of communication over the virtual graph, and each round of communication over the virtual graph can be simulated in $O(k)$ rounds of communication over the graph $G$. This step thus takes $O(k \log^* n)$. Since every deleted subgraph at this step is a direct neighbor of a subgraph that is part of the MIS, the distance to the nearest surviving subgraph is at most $O(k)$.
    
    Every subgraph still left in $\fR$ has no neighboring subgraph in the contention digraph.
\end{proof}

\section{Application: Faster \texorpdfstring{$\Delta$}{Δ}-Coloring in LOCAL}
\label{sec:detresults}
In this section, we will make use of \cref{thm:rulingsubgraphs} to design faster algorithms for $\Delta$-coloring.

\subsection{Finding Nice LDCCs}
\label{sec:nice-ldccs}

Our first step towards our improved $\Delta$-coloring algorithm is to discover useful subgraphs within moderate distance of every node in the graph. Intuitively, each of these subgraphs is always easy to color with colors from the set $[\Delta]$, even if constrained by a partial coloring outside it. They are also chosen to be relatively sparse, i.e., to contain as few nodes and edges as possible while being easy to color in the sense presented before.
We refer to these graphs as \emph{nice LDCCs}, which are DCCs with some additional properties (see \cref{def:dcc,def:ldcc,def:nice-ldcc}). A DCC is a $2$-connected subgraph that is neither a clique nor an odd cycle. LDCCs additionally exclude even cycles. Nice LDCCs are LDCCs with an additional sparseness property. We extend a proof by \cite{GHKM_dc21} to show that every node either has a node of degree $< \Delta$ or an LDCC in its distance-$2\log_\Delta n$ neighborhood. We then show how such subgraphs contain a sparse subgraph with the same properties. 

As the proof of this second lemma is relatively straightforward and lengthy, we only sketch it here with the full proof deferred to \cref{sec:nldcc-proof}.

\paragraph{Finding Locally Degree-Choosable Components}
Just like DCCs, we can show that the absence of LDCCs and of nodes of degree $<\Delta$ in the distance-$k$ neighborhood of a node implies some expansion in that neighborhood. The following lemma follows from easy arguments and \cref{lem:DCCorexpand} due to~\cite{GHKM_dc21}. Intuitively, we observe that cycles of length $>3$, both even and odd, are generally good for expansion.

\begin{lemma}
    \label{lem:LDCCorexpand}
    For every node $v$ in a graph $G$, there exist either a node of degree less than $\Delta$ at distance at most $k$ from $v$ or an LDCC of radius at most $k$ in the distance-$k$ neighborhood of $v$, or there are at least $(\Delta-1)^{\floor{k/2}}$ nodes at distance $k$ of $v$.
\end{lemma}
\begin{proof}
    \Cref{lem:DCCorexpand} states the same result for DCCs. Note that the only 2-connected subgraphs which are DCCs but not LDCCs are cycles of even length.

    Consider a node $v$, and suppose its distance-$k$ neighborhood is free of LDCCs and nodes of degree $< \Delta$.
    That is, the block graph of its distance-$k$ neighborhood is a bipartite tree where the blocks are all cycles and cliques. Now, consider every block that is a cycle. Each of them has one node $u$ that is closest to $v$ in the graph $G$, and two nodes $u',u''$ which are at distance $\dist(u,v)+1$ from $v$.
    Let us modify the distance-$k$ neighborhood of $v$ as follows: while it contains a cycle of size $>3$, we consider one such cycle $C$, consider $u$, $u'$ and $u''$ the nodes closest to $v$ in that cycle, remove from the graph all other nodes from $C$, and add an edge between $u'$ and $u''$
(which preserves their degrees). For every deleted node, we also delete all nodes whose shortest path to $v$ went through that node. From the perspective of the block graph, we replace each cycle of size $>3$ by one of size $3$, delete the edges between this cycle and the nodes that are no longer part of it, and remove all blocks and vertices that are no longer connected to the node $v$ as a result.

    Applying this transformation to the distance-$k$ neighborhood of $v$ is another valid distance-$k$ neighborhood with a lower or equal number of nodes, whose degrees did not change, and which contains no cycles of size $>3$ and no $2$-connected graphs other than cliques. We know from  \cref{lem:DCCorexpand} that this new neighborhood contains at least $(\Delta-1)^{\floor{k/2}}$ nodes, which implies the same lower bound for the initial number of nodes before the transformation.
\end{proof}

\begin{corollary}
    \label{lem:exists-ldcc}
    In a graph $G$ of maximum degree $\Delta \geq 3$, the distance-$2 \log_{\Delta-1} n$ neighborhood of every node $v$ contains either a node of degree less than $\Delta$ or an LDCC.
\end{corollary}

\paragraph{Extracting Nice LDCCs from LDCCs}

As an ingredient for our improved $\Delta$-coloring algorithms, we will show that any LDCC $H$ contains an LDCC $H'$ with a number of nodes that is linear in the diameter of $H$, and almost exclusively nodes of degree $2$ (in the induced subgraph).
This will be very useful for applying \cref{thm:rulingsubgraphs}, as it guarantees that we can choose a small value for $k$ in \cref{thm:rulingsubgraphs}. The extra condition about the degree distribution makes coloring those subgraphs a breeze in later parts of the algorithm.

\begin{restatable}{lemma}{niceLDCClemma}[Small induced nice LDCCs]
    \label{lem:sizeLDCC}
    Let $H$ be an LDCC of strong\footnote{By \emph{strong} diameter, we mean the diameter inside of $H$, not $G$. Note $H$ has diameter $k \geq 2$ since it is not a clique.} diameter $k$. Then there exists an induced subgraph $H'$ of $H$ with the following properties:
\begin{enumerate}
        \item $H'$ is an LDCC,
        \item $H'$ contains at most $4k$ nodes,
        \item $H'$ contains exactly two nodes of degree $3$, while all of its other nodes have degree $2$.
    \end{enumerate}
\end{restatable}

\begin{proof}[Proof sketch]
    The $2$-connectivity of the LDCC implies that it contains an induced cycle. The fact that the LDCC is not a cycle implies that there exists at least one node outside this cycle, which we must be able to reach from two disjoint paths from the cycle. The bound on the diameter of the LDCC guarantees that the cycle we find as well as the paths to a node outside the cycle must be of size $O(k)$. The fact that we can avoid extraneous edges between e.g.\ two nodes on our induced cycle follows from the intuitive idea that such chords simply allow us to find a smaller cycle within the nodes we selected, with no chords.
\end{proof}

\begin{corollary}[\Cref{lem:sizeLDCC,lem:exists-ldcc}]
    \label{cor:nice-ldcc}
    In a graph $G$ of maximum degree $\Delta \geq 3$, the distance-$2\log_{\Delta-1} n$ neighborhood of every node $v$ contains either a node of degree less than $\Delta$ or a nice LDCC of size at most $16\log_{\Delta-1} n$. Also, every minimal LDCC is a nice LDCC.
\end{corollary}

The full proof of \cref{lem:sizeLDCC} is deferred to \cref{sec:nldcc-proof}.

\subsection{Coloring Nice LDCCs}
\label{sec:coloring-nldcc}

Nice LDCCs have the property that for an appropriate choice of stem (\cref{def:natural-orchid}), 
once we have fixed a coloring around their stem, we can easily guarantee that the stem contains a \emph{flexible node}, that is, a node with two neighbors of the same color. More precisely, we can ensure the presence of a flexible node in a nice LDCC by coloring at most one node in its stem, which leaves the uncolored part of $H$ connected from the $2$-connectivity of LDCCs.
Given a graph $G$ that is fully colored outside of a family $\fR$ of nice LDCCs and whose contention digraph is empty, this allows us to extend the coloring to the nice LDCCs by coloring nodes in order of decreasing distance to the nearest flexible node in the subgraph induced by uncolored nodes.

We define in \cref{def:natural-orchid} the \emph{natural orchids} that we consider when handling nice LDCCs in our $\Delta$-coloring algorithm (\cref{alg:deltacoloringlocal}).
In particular, these are the orchids that we use when we compute a ruling subgraph family using \cref{alg:rulingsubgraphs} (\RulingSubgraphs).
We show that this choice of orchids (and more importantly, stems) ensures that once nodes in the orchid outside a nice LDCC are colored, the nice LDCC can create a flexible node in itself while maintaining connectivity (\cref{lem:flexible-node-LDCC}). This allows for fast coloring of a family of nice $\Delta$-Extendable Components given that their orchids do not overlap (\cref{lem:coloring-nldcc}).
The natural orchids we define are $(4,1)$-orchids for a nice LDCC.

\begin{definition}[Natural root, stem, orchid]
\label{def:natural-orchid}
    Let $H$ be a nice LDCC. Its \emph{natural root} is its node of degree $3$ of highest ID. Its \emph{natural stem} if $H$ is the natural root together with its $3$ neighbors in $H$. The \emph{natural orchid} of $H$ is the orchid of radius $1$ around the natural stem.

    For $H$ a single node of degree $<\Delta$, its root, stem and orchid are simply this single node.
\end{definition}

\begin{figure}[ht]
    \centering
    \begin{tabular}{*{3}{c}}
    \includegraphics[width=0.25\linewidth,page=4]{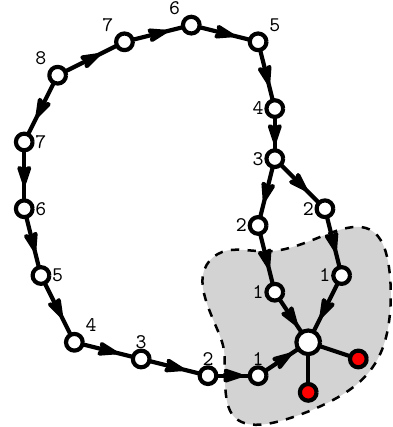}
    &
    \includegraphics[width=0.25\linewidth,page=7]{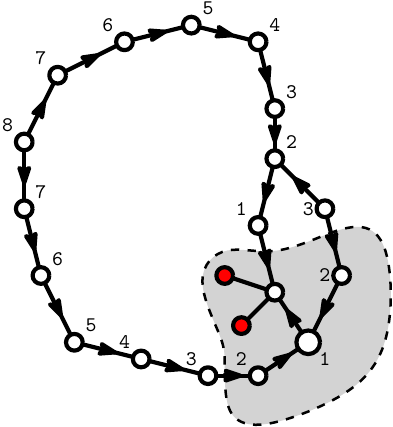}
    &
    \includegraphics[width=0.25\linewidth,page=10]{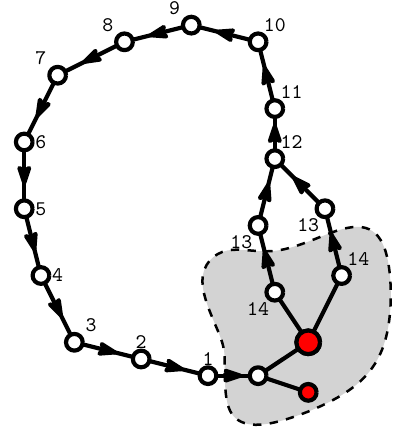}
    \end{tabular}
    \caption{Three possibilities for a nice LDCC's flexible node: (1) its root node is flexible; (2) another node in the stem is flexible; (3) a neighbor of the root in the stem is made flexible by coloring the root with a color that is used for this neighbor. In all three cases, nodes are colored in order of decreasing distance to the flexible node.}
    \label{fig:ldcc_coloring}
\end{figure}

\begin{lemma}
\label{lem:flexible-node-LDCC}
    Let $H$ be a nice LDCC that is fully uncolored while nodes in its natural orchid outside $H$
are colored. Then either its stem already contains a flexible node, or a flexible node can be created by coloring the root.\end{lemma}
\begin{proof}
    Suppose that the stem does not contain a flexible node. Consider the natural root $r_H$ of the nice LDCC and let $v \in H \cap N(r_H)$ be a neighbor of $r_H$ in the stem of degree $2$, guaranteed to exist from the fact that exactly two nodes have degree $3$ in a nice LDCC. $v$ and $r_H$ are part of the stem, and thus not flexible. Hence, the $\Delta-3$ neighbors of $r_H$ outside $H$ are colored with $\Delta -3$ distinct colors, and the $\Delta-2$ neighbors of $v$ outside $H$ are colored with $\Delta -2$ distinct colors. Let $c$ be a color that appears in the neighborhood of $v$ but not $r_H$: coloring $r_H$ with $c$ makes $v$ flexible.
\end{proof}
    
\begin{lemma}
\label{lem:coloring-nldcc}
    Let $\fR$ be a family of nice {\LDelEC}s s.t.\ $\card{H} \leq k$ for each $H \in \fR$ whose contention digraph is empty.
    Then the nodes of $\fR$ can be $\Delta$-colored in 
    $k+\Tlayer(n,\Delta,k)$ rounds.
\end{lemma}

\begin{proof}
    Recall that nice {\LDelEC}s are nice LDCCs and nodes of degree $< \Delta$ (\cref{def:nice-extendable}).
    Each nice {\LDelEC} creates a flexible node in its stem, if it does not have one -- note that a node of degree $<\Delta$ is already flexible. Consider then the graph $G'$ induced by the nodes of the nice {\LDelEC}s, without the roots that were colored to create a flexible node -- $G[\bigcup_{H \in \fR} V(H) \setminus \set{r_H}] \subseteq G' \subseteq G[\bigcup_{H \in \fR} V(H)]$. Perform a multi-origin BFS in $G'$ with the flexible nodes as origins, and label the nodes in $G'$ by their distance to the nearest flexible node in $G'$. Since each nice {\LDelEC} $H$ is $2$-connected, contains a least one flexible node, and contains $k$ nodes, every node in $H$ still has a path to a flexible node of length $< k-1$ even if the root $r_H$ of $H$ is not a vertex of $G'$. Let layer $i$ contain all the nodes at distance $i$ from the nearest flexible node in $G'$ for each $i \in \set{0,\ldots,k-1}$. Each node $v$ in a layer of index $i > 0$ has at least one neighbor in layer $i-1$, which means that its up degree in $G'$ is $\widehat{\delta}(v) \leq \Delta-1$ and while it has access to at least $\widehat{\delta}+1$ colors $\subseteq [\Delta]$. The flexible nodes in layer $0$ have at least $\widehat{\delta}+1$ free colors $\subseteq [\Delta]$ from the fact that a color is repeated in their neighborhoods.
\end{proof}
    
For concrete bounds on the complexity of this operation, see the current state-of-art complexities for coloring layered graphs in \cref{lem:d1lc_local,lem:layered_graphs}.

\subsection{The Complexity of \texorpdfstring{$\Delta$}{Δ}-Coloring Parameterized by \texorpdfstring{$n$}{n} and \texorpdfstring{$\Delta$}{Δ}}
\label{sec:detlocal}

The goal of this section is to prove \cref{thm:detLOCAL}, which we restate here for convenience.

\deterministicLOCAL*

The algorithm we will use to obtain \cref{thm:detLOCAL} is \cref{alg:deltacoloringlocal}. Our next proof shows that it computes a proper $\Delta$-coloring in the runtime stated in \cref{thm:detLOCAL}.
While the description of \cref{alg:deltacoloringlocal} omits some details regarding how to perform some of the mentioned steps, these details will be filled in later in the appropriate places.

\begin{algorithm}[ht!]
\caption{$\Delta$-coloring in deterministic LOCAL}
\label{alg:deltacoloringlocal}
\begin{algorithmic}[1]
\State Each node $v$ finds a nice locally $\Delta$-extendable component in $N^{2\log_{\Delta-1} n}(v)$ and selects the smallest one, breaking ties arbitrarily.

\Statex Let $\fR$ be the family of subgraphs formed by all the {\NLDelEC}s selected in the previous step. For each $H \in \fR$, take its natural orchid $F(H)$ (\cref{def:natural-orchid}).

\State Compute $\fR'\subseteq \fR$ by applying \cref{alg:rulingsubgraphs} on the family $\fR$ with their natural orchids, with parameters $r = 1$ and $k = 16 \log_{\Delta-1} n +4$.

\Statex $\fR'\subseteq \fR$ satisfies the properties stated by \cref{thm:rulingsubgraphs} (for the respective $r, k$).

\State Compute for each node $v \in V \setminus \bigsqcup_{H \in \fR'} H$ (i.e., outside all members of $\fR'$) its distance to the nearest member of $\fR'$, defining layers
\State Color all nodes from $V \setminus \bigsqcup_{H \in \fR'} H$ using an algorithm for coloring layered graphs,
using colors $1$ to $\Delta$.
\State Each  member $H$ of $\fR'$ creates a flexible node in its stem, if it does not already have one.
\State In the subgraph induced by uncolored nodes, compute the distance of each uncolored node to the nearest flexible node.
\State Color the remaining uncolored nodes with an algorithm for coloring layered graphs, using colors $1$ to $\Delta$.
\end{algorithmic}
\end{algorithm}

\begin{proof}[Proof of \cref{thm:detLOCAL}]
    We start by proving that each node $v$ is at distance at most $2\log_{\Delta-1} n$ from a subgraph in the family $\fR$, and that all members of $\fR$ are nice locally $\Delta$-extendable components (\NLDelEC) of size at most $16 \log_{\Delta-1} n + 4$.

    Fix an arbitrary node $v \in V(G)$.
    By \cref{lem:exists-ldcc}, the $(2\log_{\Delta-1} n)$-hop neighborhood $N^{2\log_{\Delta-1} n}(v)$ of $v$ contains a node of degree less than $\Delta$ or an LDCC.
    In the former case, we have a {\NLDelEC} of size $1$ within distance $2\log_{\Delta -1}$ from $v$. 
In the latter case, recall that each node selected an LDCC of smallest size. Since every non-nice LDCC contains a nice LDCC as induced subgraph by \cref{lem:sizeLDCC}, selecting a minimal LDCC guarantees it to be a nice LDCC. Furthermore, since all LDCCs are found within the distance-$2\log_{\Delta-1} n$ neighborhood of a node, each minimal LDCC found is a subgraph of an LDCC of strong diameter $\leq 4\log_{\Delta-1} n+1$, which necessarily contains a nice LDCC of at most $16\log_{\Delta-1} n +4$ nodes by \cref{lem:sizeLDCC}. This first step of the algorithm is done where each node finds the nearest member of $\fR$ is done in $O(\log_{\Delta - 1} n) = O(\log_{\Delta} n)$ rounds (since $\Delta \geq 3$).   

    Each {\NLDelEC} of size $1$ has its unique node as root, while {\NLDelEC} that are nice LDCCs use their node of degree $3$ of highest ID as root. For each selected {\NLDelEC} $H \in \fR$, its stem is the inclusive neighborhood of its root in $H$.
    \Cref{alg:rulingsubgraphs} is then applied with the family $\fR$, whose members have size at most $k=16\log_{\Delta-1}n+4$, have stems of size $t\leq 4$, and with $d=1$. Let us denote by $\Trsg(n,\Delta,k,t,d)$ its running time, and $\fR' \subseteq \fR$ the smaller family of {\NLDelEC} it produced. By \cref{thm:rulingsubgraphs},
    \[
    \Trsg(n,\Delta,k,t,d) \leq O(\log \Delta + \log_\Delta n \log^*n)\ .
    \]

    The same \cref{thm:rulingsubgraphs} guarantees that every member $H \in \fR$ of the original family is at distance at most $O(\log \Delta + \log_\Delta n \log^* n)$ from a member of the smaller family $\fR'$.
    This means that every node in the graph is at distance at most $h \leq O(\log \Delta + \log_\Delta n \log^* \log_\Delta n)$ from some $H' \in \fR'$.
    
    Partitioning the nodes of outside the family $\fR'$ into layers according to their distance to $\fR'$, we obtain a layered graph instance where the maximum up degree is $\widehat{\Delta} \leq \Delta - 1$ and with $h\in O(\log_\Delta n \log^* \log_\Delta n)$ layers. We use $\Tlayer(n,\Delta,h)$ rounds to colors those nodes.

    Extending the coloring to the nodes of $\fR'$ is then done in $k+\Tlayer(n,\Delta,k)$ rounds, as shown in \cref{lem:coloring-nldcc}, by first ensuring the existence of a flexible node in each nice {\LDelEC} $H \in \fR'$ and coloring the other nodes of the {\LDelEC}s with an algorithm for coloring layered graph.
    
    Observing that the complexity of coloring layered graphs is non-decreasing w.r.t.\ the bound on the number of layers, and at least linear in this number of layers (from the complexity of coloring a path), the overall complexity is of order $O( \Trsg(n,\Delta,k,t,d) + \Tlayer(n,\Delta,h))$ where $k \in O(\log_\Delta n)$, $t\in O(1)$, $d \in O(1)$, and $h \in O(\log \Delta + \log_\Delta n\log^*n)$. In particular, using the algorithm by Ghaffari and Kuhn for coloring layered graphs (\cref{lem:layered_graphs}), we get a complexity of
    \[
     O(\log^4 \Delta + \log^2\Delta \cdot \log n \log^* n)
    \ .\qedhere\]
\end{proof}

\subsection{Optimizing for the High-\texorpdfstring{$\Delta$}{Δ} Regime}
\label{sec:high-delta-optimization}

While on constant-degree graphs the algorithm of \cref{thm:detLOCAL} produces a $\Delta$-coloring in $O(\log n \log^*n)$ rounds, which improves almost quadratically on the $O(\log^2 n)$ state-of-the-art algorithms of \cite{panconesi95delta,GHKM_dc21,ghaffari2021deterministic}, the attentive reader might notice that our algorithm performs worse than some combinations of prior works presented in \cref{sec:complexity-in-n} in higher regimes of $\Delta$ -- more precisely, when $\log \Delta > \log^c n$ for some constant $c$.
In this section, we explain how minor modifications to our algorithm suffice to achieve complexities that also improve on the complexities from \cref{sec:complexity-in-n}, albeit without providing an improvement for the complexity as a pure function of $n$ stronger than saving some $\log \log^{O(1)} n$ terms in some regimes of $\Delta$.

The high dependency in $\Delta$ in \cref{thm:detLOCAL} has two origins: first, an increase of $O(\log \Delta)$ in the distance to the nearest subgraph from {\RulingOrchids}; second, the use of the algorithm by Ghaffari and Kuhn~\cite{ghaffari2021deterministic} for coloring layered graphs, which introduces a $O(h \log^3 \Delta)$ term to the complexity, where $h$ is the number of layers of the graph (see \cref{lem:layered_graphs}). Both can be swapped for faster procedures in the high $\Delta$ regime.

First, the algorithm for computing ruling sets used in {\RulingOrchids} can be replaced by the recent algorithm by Ghaffari and Grunau~\cite{GG_focs24}, which computes an $(2,O(\log \log \Delta))$-ruling set in $\Otilde(\log n)$ rounds. This increases the complexity of  {\RulingOrchids} and thus \cref{alg:rulingsubgraphs} to $\Otilde(\log n)$, but in exchange reduces the distance of every node to the nearest member of the family output by \cref{alg:rulingsubgraphs} to a lower distance $h$ of
$O(\log \log \Delta + \log_\Delta n \log^* n)$. When still using the algorithm by Ghaffari and Kuhn~\cite{ghaffari2021deterministic} for coloring layered graphs, we get an algorithm for $\Delta$-coloring of total complexity
\[
\Otilde(\log n) + O(\log^3\Delta \cdot \log \log \Delta + \log^2 \Delta \cdot \log n \log^* n)\ .
\]

If we instead color our layered graphs one layer at a time using the $\Otilde(\log^{5/3} n)$-round algorithm by Ghaffari and Grunau~\cite{GG_focs24} (see \cref{lem:d1lc_local}), we obtain an algorithm of complexity
\[
\Otilde(\log^{8/3} n / \log \Delta)
\]
matching the complexity of the algorithms in \cref{thm:sota-pure-n-dependency} in their high order terms.

\section{Consequences in the Randomized Setting}
\label{sec:randomized}

Our techniques in the deterministic setting from previous sections also imply faster randomized algorithms. In this section, we present the faster randomized algorithms that we obtain by combining algorithmic approaches from prior work with our newer techniques. Our algorithm essentially follows the approach of \cite{GHKM_dc21}, enhanced by our algorithm for computing ruling subgraphs.

This section is devoted to proving the following theorem:

\randomizedLOCAL*

More precisely, we show that \cref{alg:randLOCAL} computes a $\Delta$-coloring in the stated runtimes, w.h.p. We get two different runtimes from this single algorithm based on which subroutine we use to color the layered graph in Step 7 of \cref{alg:randLOCAL}.

\begin{algorithm}[htb!]
\caption{$\Delta$-coloring in randomized LOCAL}
\label{alg:randLOCAL}
\begin{algorithmic}[1]
\Statex Set $d \gets 20 \log_\Delta \log n + 80$.
\State Find all nice locally $\Delta$-extendable subgraphs ({\NLDelEC}) of size at most $5d$. \Comment{$ \log_\Delta \log n$}

\Statex Let $\fR$ be the family of those subgraphs and $H$ the graph induced by their vertices.
\State Apply \cref{alg:rulingsubgraphs} to get a ruling family $\fR' \subseteq \fR$ over $H$.
\State Each node at distance $5d$ or more from the subgraph $H$ participates in $T$-node sampling.

\Statex Let $S$ be the union of $V(H)$ and the set of sampled $T$-nodes, and $H'$ the graph induced in $G$ by uncolored vertices.
\State Each node $v$ computes $\ell_v = \dist_{H'}(v,S)$, the length of its shortest uncolored path to either a $T$-node or a subgraph from the family $\fR'$.
\Comment{$ \log_\Delta \log n \cdot \log^* n$}
\State Partition nodes into layers according to their $\ell_v$.
\State Color the layered graph layer by layer, in order of decreasing distance.
\State Color $T$-nodes and the subgraphs of the Locally $\Delta$-extendable subgraph family $\fR'$.
\end{algorithmic}
\end{algorithm}

The  main idea behind the speedup we get in randomized \LOCAL compared to deterministic \LOCAL relies in the sampling of \emph{$T$-nodes} (\cref{alg:Tnodes}). $T$-nodes are sampled such that every two $T$-nodes are at some minimum constant distance $b$ -- $34$ in our case, $15$ in the earlier work by Ghaffari, Hirvonen, Kuhn, and Maus~\cite{GHKM_dc21}.
After their sampling, each $T$-node selects a pair of unconnected neighbors (which are said to be \emph{marked}), which are colored using the same color, e.g., $1$.

\begin{algorithm}[htb!]
\caption{$T$-nodes Sampling, with conflict distance $b$}
\label{alg:Tnodes}
\begin{algorithmic}[1]
\State Each node selects itself with probability $1/\Delta^{b}$
\State Each selected node with a selected node in its distance-$b$ neighborhood unselects itself
\State Each remaining selected node becomes a $T$-node; it picks a pair of unconnected neighbors uniformly at random and marks them.
\State Each marked node colors itself $1$.
\end{algorithmic}
\end{algorithm}

To see why this is useful, suppose that we extended the coloring so that all nodes except $T$-nodes are colored -- importantly, without changing the coloring of the marked nodes. As each $T$-node has at least $2$ neighbors with the same color, it sees at most $\Delta-1$ distinct colors in its neighborhood. Each $T$-nodes can thus easily pick a free color for itself in $[\Delta]$, and the coloring is easily extended to $T$-nodes. Additionally, $T$-nodes can be used to partition uncolored nodes into layers according to their \emph{uncolored distance} to the nearest $T$-node, i.e., the length of their shortest uncolored path to their nearest $T$-node. The uncolored distance defines a convenient order in which to color the graph, as does the distance to the nearest locally $\Delta$-extendable component in the deterministic algorithms.

The crux of the argument is then to show that each node in the graph is within short uncolored distance of a $T$-node. Alas, this cannot be guaranteed in all parts of the graph. However, it can be shown that in areas of the graph that expand (i.e., where the distance-$r$ neighborhoods have size $\Delta^{\Omega(r)}$, for $r \in \Theta(\log_\Delta \log n)$), nodes are likely to have a short uncolored path to a $T$-node. Areas of the graph which do not expand necessarily contain locally $\Delta$-extendable subgraphs, which can be exploited by nodes in those parts of the graph to color themselves efficiently, as in the deterministic algorithms.

The properties we need regarding the outcome of the sampling of $T$-nodes where shown in the work of Ghaffari, Hirvonen, Kuhn, and Maus~\cite{GHKM_dc21}. We summarize the technical results we need in the next series of lemmas.

\begin{lemma}[{\cite[combination of Lemmas 5 and 7]{GHKM_dc21}}]
    \label{lem:lemm5and7}
    Let $G$ be a $\Delta$-regular graph without any DCCs of radius at most $r$. After applying $T$-node sampling with conflict distance 15, in the graph induced by unmarked nodes, each node $v$ either has a $T$-node at distance $r$ or less or has at least $\Delta^{\floor{r/5}}$ nodes at distance $r$.
\end{lemma}

\Cref{lem:lemm5and7} guarantees that nodes without a DCC or node of degree $< \Delta$ within some distance $r \geq \log_\Delta \log n + 80$ have an unmarked path to a $T$-node or to at least $\Delta^{16} \log n$ nodes at distance $O(\log_\Delta \log n)$, regardless of which nodes get marked by the $T$-node sampling process. In particular, conditioning on some arbitrary set of random choices within distance $r-16$ of a node $v$, the random choices made at the $\Omega(\Delta^{16} \log n)$ nodes at distance between $r-16$ and $r$ from $v$ have a high probability of creating a path between $v$ and a $T$-node. This is because for each node $u$ reachable from $v$ through an unmarked path of length $r$, where the distance-16 neighborhood of $u$ has not yet made its random choice in $T$-node sampling, there is a probability at least $\Omega(1/\Delta^{16})$ that the path ends up leading to a $T$-node through unmarked nodes. That is, there is a probability at least $\Omega(1/\Delta^{15})$ that $u$ gets selected as $T$-node, and a similar probability that a neighbor of $u$ is selected as $T$-node. Either $u$ or one of its neighbor has a probability at least $\Omega(1/\Delta)$ of marking a pair of neighbors with does not remove the unmarked path between $u$ and $v$. Analyzing a single node, this yields \cref{lem:lemm17}.

\begin{lemma}[{\cite[Lemma 17, adapted to non-constant $\Delta$]{GHKM_dc21}}]
    \label{lem:lemm17}
    Let $G$ be a graph without any DCCs of radius at most $r \geq 32$. 
    Let $v$ and $u$ be nodes such that there is an unmarked unselected path from $u$ to $v$.
    
    After applying the marking process removing nodes within distance $l$ from each other.
    Then $u$ or one of its neighbors becomes a $T$-node with  probability $\Theta(1/\Delta^{l})$. 

    Notably, if $l=15$ , then $u$ or one of its neighbors becomes a $T$-node with  probability $\Theta(1/\Delta^{15})$.
\end{lemma}

Combining \cref{lem:lemm5and7} and \cref{lem:lemm17}, gives the following result:

\begin{lemma}[{\cite[Lemma 18]{GHKM_dc21}}]
    \label{lem:lemm18}
    Let $G$ be a $\Delta$-regular graph without any DCCs of radius at most $r \geq 20 \log_\Delta \log n + 80$. After applying $T$-node sampling with distance 15, each unmarked node $v$ has an unmarked path of length $O(\log_\Delta \log n)$ to a $T$-node, w.h.p.
\end{lemma}

We now show the following lemma giving us the same result for $\Delta$-regular graph without any LDCCs of radius at most $r \geq 20 \log_\Delta \log n + 80$, and with a higher distance of $34$ between $T$-nodes compared to 15 in~\cite{GHKM_dc21}.

\begin{lemma}
    \label{lem:TnodesLDCC}
    Let $G$ be a $\Delta$-regular graph without any LDCCs of radius at most $r \geq 40 \log_\Delta \log n + 80$. After applying $T$-node sampling with distance $34$, each unmarked node $v$ has an unmarked path of length $O(\log_\Delta \log n)$ to a $T$-node, w.h.p.
\end{lemma}

\begin{proof}
    We are going to use a similar construction as the one in \cref{lem:LDCCorexpand}. As in that lemma, we start by considering a graph without any DCCs of radius at most $r$. We consider a node $v$, and the BFS tree starting from $v$. 
    We first start by sampling $T$-nodes with conflict distance $b=34$ in $G$. 

    Then, applying the same construction as in \cref{lem:LDCCorexpand}, we modify the distance-$k$ neighborhood of $v$: while it contains a cycle of size $>3$, we consider one such cycle $C$, then $u$, $u'$ and $u''$ the nodes closest to $v$ in that cycle with $u$ the closest to $v$. We remove from the graph all other nodes from $C$, and add an edge between $u'$ and $u''$. For every deleted node, we also delete all nodes whose shortest path to $v$ went through that node. 
    If one of the neighbors of $u'$ or $u''$ was a $T$-node then we choose $u'$ to become a $T$-node and have it select $u$ and $u''$.
    Since the block graph made of the cycles and cliques is a tree, no child of $u$ will be reachable from $v$ using only unselected nodes. Doing this operation ensures that strictly less nodes will be reachable from $v$ using only unselected untouched nodes with this choice of $T$-nodes compared to the initial $T$-node sampling.

    The distance-$k$ neighborhood of $v$ now has a lower or equal number of nodes as in $G$, with degrees all equal to $\Delta$, and it neither contains cycles of size $>3$ nor $2$-connected subgraphs other than cliques. The distance between two $T$-nodes is at least $15$. Indeed, we have added edges between nodes which were at distance $2$, decreasing the distance between nodes by at most half. Also, the position of $T$-nodes could have change from one node to another neighboring node. Which implies that the distance between two $T$-nodes is at least $34/2-2=15$

    We now get from \cref{lem:lemm5and7} that $v$ can reach $\Delta^{r/5}$ nodes at distance $r$ using unselected nodes or can reach a $T$-node at distance $r$ in the construction. Since any reachable $T$-node in the construction can be reached in $G$ and any node at distance $r$ reachable using unselected nodes in the construction can also be reached using the same nodes in $G$, we obtain a generalization of \cref{lem:lemm5and7} for the LDCCs.
    
    \Cref{lem:lemm17} gives us that if there is $u$ such that there is a path of unmarked and unselected nodes from $v$ to $u$, then after applying the marking process, and removing nodes, $u$ or one of its neighbors becomes a $T$-node with probability $\Theta(1/\Delta^{34})$. 
    
    We get \cref{lem:TnodesLDCC} by combining these two results.
\end{proof}

\begin{proof}[Proof of \cref{thm:randLOCAL}]

    Discovering locally $\Delta$-extendable subgraphs of size $5d \in O(\log_\Delta \log n)$ only takes $O(\log_\Delta \log n)$, as well as each node learning whether it has such a subgraph within distance $5d$.

    A node v without a locally $\Delta$-extendable subgraph at distance $5d$ participates in $T$-node sampling. If all nodes at distance $5d$ from $v$ participate in $T$-node sampling, by \cref{lem:lemm18}, $v$ has an unmarked (i.e., uncolored) path to a $T$-node of length $O(d)$, w.h.p. If $v$ does not have such an unmarked path, then $v$ must have an unmarked path of length $O(d)$ to a node which did not participate in $T$-node sampling. Since nodes that do not participate in $T$-node sampling have a path of length $O(d)$ to a locally $\Delta$-extendable component, every node in the graph has an unmarked path of length $O(d)$ to either a $T$-node or a locally $\Delta$-extendable component.

We use \cref{thm:rulingsubgraphs}, but with some changes to \RulingOrchids\, where instead of computing a $(2,\log \Delta$)-ruling set on line $1$, the algorithm will do a $(2,O(\log \log n))$-ruling set which can be done in $O(\log \log n)$ rounds in randomized (\cref{lem:ruling-set}). Then \RulingOrchids\ has a runtime of $O(\log \log n)$-rounds in our current setting since the conflict graph we use in the algorithm has constant diameter.
    Consider now the runtime of computing the family $\fR'$ of ruling locally $\Delta$-extendable subgraphs (i.e., applying \cref{thm:rulingsubgraphs}). With the size of the subgraphs given as input, this takes $O(d \log^ *n +\log \log n) = O(\log_\Delta \log n \log^ *n + \log \log n)$ rounds. Every node in the graph is within distance $O(\log_\Delta \log n \log^ *n + \log \log n)$ of a $T$-node or a locally $\Delta$-extendable subgraph in $\fR'$.

    At this point, nodes partition themselves into $O(\log_\Delta \log n \log^ *n + \log \log n)$ layers, which we color one by one. This is the most expensive part of the algorithm: after this step, coloring $T$-nodes only takes $O(1)$ rounds, and coloring the locally $\Delta$-extendable subgraph family $\fR'$ 
    amounts to coloring a layered graph with a smaller number of layers,  $O(\log_\Delta \log n)$.
    Depending on which algorithm we use for the two steps involving coloring layered graphs, we obtain the two stated complexities.
    
    Computing an auxiliary $O(\Delta^2)$-coloring using Linial's algorithm (\cref{lem:linial}), and applying the $O(\sqrt{\Delta \log\Delta})$ algorithm for degree+1 list coloring by Maus and Tonoyan \cite{MausTonoyan_dc22} to color each layer (see \cref{lem:d1lc_local}), we obtain a complexity of $O(\sqrt{\Delta \log\Delta} (\log_\Delta \log n \log^ *n + \log \log n))$.

    If instead we use the $\Otilde(\log^{5/3} \log n)$ randomized algorithm for degree+1 list coloring by \cite{HalldorssonKNT22near,GG_focs24} (see \cref{lem:d1lc_local}), we obtain a complexity of $\Otilde(\log^{8/3} \log n)$.    
\end{proof}

\urlstyle{same}
\bibliographystyle{alpha}
\bibliography{arxiv_v1_refs}

\appendix

\section{Complexity Upper Bounds Achievable by Combining Prior Work}
\label{sec:complexity-in-n}

In this appendix, we show how to combine prior work to obtain improved complexity bounds for deterministic $\Delta$-coloring in the LOCAL model.
We emphasize that our new results also supersede the results from this section.

\begin{restatable}{theorem}{DeltaColFromGG}[Consequence of \cite{GG_focs24} and \cite{ghaffari2021deterministic}]
    \label{thm:deltaColGG}
    There exist algorithms for $\Delta$-coloring graphs of maximum degree $\Delta \geq 3$ in $\Otilde(\log^{8/3} n / \log \Delta)$ and $\Otilde(\log^2 n / \log \Delta + \log n \log^2 \Delta)$ rounds of \LOCAL.
\end{restatable}

\begin{proof}
    Ghaffari and Grunau~\cite{GG_focs24} gave an algorithm for MIS of complexity $\Otilde(\log^{5/3} n)$ and an algorithm for $(2,O(\log \log \Delta))$-ruling set of complexity $\Otilde(\log n)$. We first give an algorithm making use of the MIS algorithm alone, before one that makes use of the ruling set algorithm as well as the algorithm for coloring layered graph by Ghaffari and Kuhn~\cite{ghaffari2021deterministic}.
    
    Consider the conflict graph over all $\Delta$-extendable subgraphs of size $O(\log_\Delta n)$, and recall that every node in the original graph is guaranteed to be within distance $O(\log_\Delta n)$ of one such $\Delta$-extendable subgraph. Computing an MIS over the conflict graph can be done in $\Otilde(\log^{8/3} n / \log \Delta)$, using the MIS algorithm of \cite{GG_focs24} and since a round of communication over the conflict graph can be performed in $O(\log_\Delta n)$ rounds of communication over the original graph $G$. Let $\fR$ be the subgraph family computed by the MIS algorithm. Each node in the graph is still at a distance at most $h_{\max} \leq O(\log_\Delta n)$ from this smaller set of subgraphs $\fR$.
    
    Each node $v \in G$ then computes its distance $h(v)$ to the nearest member of $\fR$. For each $x\in \set{1,\ldots, h_{\max}}$, let $V_x = h^{-1}(x)$ be the subset of nodes whose distance to the nearest member of $\fR$ is $x$. Note that every node in $V_x$ has at least one neighbor in $V_{x-1}$, and therefore, at most $\Delta-1$ neighbors in $\bigcup_{y \geq x} V_x$. Thus, we can color $V_{h_{\max}},\ldots,V_1$ one by one, using only colors from $1$ to $\Delta$, using the MIS algorithm of \cite{GG_focs24} to solve the degree+1-list coloring instance corresponding to coloring $V_x$ after all subsets $V_{x'}$ for $x'>x$ have already been colored. Since there are $O(\log_\Delta n)$ subsets $V_x$ to color, and coloring each of them can be done in $\Otilde(\log^{5/3} n)$, the complexity of coloring all subsets $V_x$ for $x$ between $1$ and $h_{\max}$ is $\Otilde(\log^{8/3} n/\log \Delta)$. 

    Finally for this first algorithm, there remains to color the subgraphs in $\fR$, which is done in $O(\log_\Delta n)$ as members of $\fR$ are $\Delta$-extendable components of diameter $O(\log_\Delta n)$. The total complexity is $\Otilde(\log^{8/3} n/\log \Delta)$.

    For the second complexity, we instead compute a $(2,O(\log \log \Delta))$-ruling set over the conflict graph of $\Delta$-extendable subgraphs of size $O(\log_\Delta n)$. This is done in $\Otilde(\log^2 n / \log \Delta)$ rounds. We then partition the nodes outside the ruling set by their distance to it, resulting in at most $h_{\max} \leq O(\log n \log \log \Delta / \log \Delta)$ layers. Using the algorithm by Ghaffari and Kuhn for coloring layered graphs, and finishing the coloring afterwards, we obtain a complexity of $\Otilde(\log^2 n / \log \Delta + h_{\max} \log^3 \Delta + \log^2 \Delta \log n) \leq \Otilde(\log^2 n / \log \Delta + \log n \log^2 \Delta)$.
\end{proof}

\begin{restatable}{theorem}{PreviousOptDeltaCol}[Optimal combination of \cite{ghaffari2021deterministic} and \cite{GG_focs24}]
    \label{thm:sota-pure-n-dependency}
    $\Delta$-coloring has complexity $\Otilde(\log^{19/9} n)$ in \LOCAL as a pure function of $n$.
\end{restatable}

\begin{proof}
    The algorithms for $\Delta$-coloring described in \cref{thm:deltaColGG} have complexity $\Otilde(\log^{8/3} n / \log \Delta)$ and $\Otilde(\log^2 n / \log \Delta + \log n \log^2 \Delta)$.

    By using the first algorithm when $\log \Delta \geq \log^{5/9} n$, and the second when $\log \Delta \geq \log^{5/9} n$, we get the claimed complexity of $\Otilde(\log^{19/9} n)$.
\end{proof}

When focusing on the dependence on $\Delta$, the prior algorithm of Ghaffari and Kuhn~\cite[Section 5.4 in arxiv version]{ghaffari2021deterministic}
achieved a complexity of $O(\log^2 \Delta \log^2 n)$ rounds.

\section{A Polylogarithmic CONGEST Algorithm for \texorpdfstring{$\Delta$}{Δ}-Coloring}
\label{sec:sota-congest}

In this section, we show how a simple combination of prior work yields an $O(\poly(\log n))$ algorithm for $\Delta$-coloring in the \CONGEST model.

\begin{theorem}[Combination of \cite{Awerbuch89}, \cite{panconesi95delta}, and \cite{ghaffari2021deterministic}]
    \label{thm:sota-det-congest}
    There exists an $O(\log^2 \Delta \log^2 n)$ deterministic algorithm for $\Delta$-coloring graphs of maximum degree $\Delta \geq 3$ in the \CONGEST model.
\end{theorem}

The theorem follows from these three results, all working in deterministic \CONGEST:
\begin{itemize}
    \item The algorithm for computing ruling sets by Awerbuch, Goldberg, Luby, and Plotkin~\cite{Awerbuch89}.
    \item A method for extending a partial $\Delta$-coloring with only a few well-spaced uncolored nodes to a full $\Delta$-coloring by Panconesi and Srinivasan~\cite{panconesi95delta}.
    \item The rounding-based algorithm for coloring layered graphs by Ghaffari and Kuhn~\cite{ghaffari2021deterministic}, which has a version working in \CONGEST.
\end{itemize}

We state them formally in the next three lemmas.

\begin{lemma}[Lemma 1 in \cite{Awerbuch89}]
    \label{lem:rulingsetCONGEST}
    There is an algorithm for computing a $(k,k\log n)$-ruling set in $O(k \log n)$ rounds of \CONGEST.
\end{lemma}

\begin{lemma}[{\cite[Lemma 1 and Theorem 1]{panconesi95delta}}]
    \label{lem:path-recoloring-congest}
    Let $G$ a connected subgraph such that $\Delta \geq 3$, $G$ is not a clique, and $G-v$ is $\Delta$-colored. Then the $\Delta$-coloring can be extended to include $v$ by recoloring a path originating from $v$ of length at most $O(\log_\Delta n)$. The can be found in $O(\log_\Delta n)$ rounds of \CONGEST.
\end{lemma}

While \cref{lem:path-recoloring-congest} only mentions the recoloring of a single vertex, it immediately implies the same result for a set of nodes with sufficient distance between its members. This is the heart of the parallel results in \cite{panconesi95delta}. Furthermore, the paths only need $O(\log_\Delta n)$ rounds to be recolored, as no symmetry breaking is needed since paths do no touch each other.

\begin{corollary}[Corollary of \cref{lem:path-recoloring-congest}]
    \label{cor:parallel-path-recoloring-congest}
    Let $G$ a connected subgraph such that $\Delta \geq 3$ and $G$ is not a clique. Let $G$ be partially $\Delta$-colored, and $S$ be its set of uncolored nodes. Assume that for all pairs of nodes $u,v \in S$, $\dist(u,v) \geq c \log_\Delta n$ for a sufficiently large constant $c$.
    Then the partial coloring of $G$ can be extended to a full $\Delta$-coloring in $O(\log_\Delta n)$ rounds of \CONGEST.
\end{corollary}

Recall that when talking about layered graphs (\cref{def:layered_graph}), $h$ refers to the number of layers, $\widehat{\delta}(v)$ to the number of neighbors that $v$ has in layers of equal or lesser index than its own, and $\widehat{\Delta}$ is a global bound on the maximum up degree.

\begin{lemma}[{\cite[Lemma 5.5 in arxiv version]{ghaffari2021deterministic}}]
\label{lem:layered_graphs_congest}
    Consider a layered graph where each node $v \in V$ receives a list of colors $L_v$ of size $|L(v)| \geq \widehat{\delta}(v)+1$ containing colors between $1$ and $C$. There exists a deterministic \CONGEST algorithm that list-colors the graph in $O(\log^2 C \cdot \log n + h \cdot \log^2 C \cdot \log \widehat{\Delta})$ rounds, using messages of at most $O(\log C + \log n)$ bits.
\end{lemma}

\begin{proof}[Proof of \cref{thm:sota-det-congest}]
    We follow a similar approach as in prior work (see, e.g., \cite[Section 5.4 in arxiv version]{ghaffari2021deterministic}). First, we compute an $(k,k\log n)$ ruling set, where $k \in \Theta(\log_\Delta n)$ where the $\Theta(.)$ hides a sufficiently large constant. This is done in $O(\log^2 n / \log \Delta)$ by \cref{lem:rulingsetCONGEST}. We partition nodes of the graph into $\Theta(\log^2 n / \log \Delta)$ layers according to their distance to the ruling set. Using \cref{lem:layered_graphs_congest}, we color the entire graph except for nodes in the ruling set in $O(\log^2 n \log^2 \Delta)$ rounds. Finally, each node in the ruling set finds a path of length $O(\log_\Delta n)$ that it can use to color itself as guaranteed by \cref{cor:parallel-path-recoloring-congest}. Finding the path, and changing the colors of its nodes as needed, is done in $O(\log_\Delta n)$ rounds.
\end{proof}

We suspect that a combination of the works mentioned in this appendix with the shattering procedure of Ghaffari, Hirvonen, Kuhn, and Maus~\cite{GHKM_dc21} that we use in our results for randomized LOCAL (\cref{sec:randomized}) could similarly yield a randomized algorithm of complexity $O(\poly(\log \log n))$.

\section{Extraction of Nice LDCCs}
\label{sec:nldcc-proof}

This appendix contains the full proof of \cref{lem:sizeLDCC}, which we only sketched in \cref{sec:nice-ldccs}. We restate its statement and the definition of nice LDCCs here.

\begin{figure}[ht]
    \centering
    \includegraphics[width=0.2\linewidth,page=1]{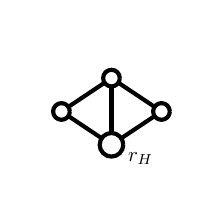}\includegraphics[width=0.2\linewidth,page=2]{ldcc_examples.pdf}\includegraphics[width=0.2\linewidth,page=3]{ldcc_examples.pdf}\includegraphics[width=0.2\linewidth,page=4]{ldcc_examples.pdf}\caption{Examples of nice LDCCs. Each of them is 2-connected, has two nodes of degree $3$, while all other nodes have degree $2$.}
    \label{fig:ldcc-examples}
\end{figure}

\defNLDCC*

\niceLDCClemma*

\begin{proof}
    We show the following, in order, about an LDCC $H$ of strong diameter $k$: that if $H$ has a clique of size $\geq 3$ as induced subgraph, it also contains a nice LDCC $H'$ of size at most $2k+2$; that if $H$ has a cycle as induced subgraph, it also contains a nice LDCC $H'$ of size at most $4k$; that every LDCC has a cycle as induced subgraph.
    The result follows: for an LDCC $H$, we consider the smallest cycle $C$ that is an induced subgraph of $H$; if $C$ has size $3$, the first argument about cliques applies; if $C$ has size $4$ or more, the second argument about larger cycles applies.

    \subparagraph*{Cliques can be extended to a nice LDCC}
    Let $K$ be a clique of maximum size $\geq 3$ in $H$.
    $H$, being an LDCC, is not a clique, so $V(H) \not \subseteq V(K)$. $H$ is also connected, implying there exists some $v \in V(H)\setminus V(K)$ s.t.\ $N(v) \cap V(K) \neq \emptyset$, and by maximality of the clique $K$, $N(v) \cap V(K) \neq V(K)$ for each such $v$.
    Consider the neighbors of $K$, $N(K) \setminus V(K)$.

    If there exists $v \in N(K) \setminus V(K)$ s.t.\ $\card{N(v) \cap V(K)} \geq 2$, then the set $U' = \set{v,u,u',u''}$ where $u,u' \in N(v) \cap V(K)$ and $u'' \in V(K) \setminus N(v)$ induces a nice LDCC of size $4$.
Otherwise, $\card{N(v) \cap V(K)} =1$ for all $v \in N(K) \setminus V(K)$: let us denote by $u_v$ the unique vertex in $N(v) \cap V(K)$ for each $v \in N(K) \setminus V(K)$.
    $\card{N(K) \setminus V(K)} \geq 2$ by the $2$-connectivity of $H$, as if $\card{N(K) \setminus V(K)}$ contained a single vertex $v$, removing $u_vv$ would disconnect $H$.

    For each $v \in N(K) \setminus V(K)$, $v' \in N(K) \setminus (V(K)\cup \set{v})$, let $P_{v,v'}$ be a shortest path from $v$ to $v'$ in $H \setminus \set{u_v}$, guaranteed to exist by $H$ being $2$-connected.
    Let $v,v'$ be two nodes generating a path $P_{v,v'}$ of minimal size for all possible pairs $v,v'$. 
    Let $t$ be its length and $w_0 = v,w_1,\ldots,w_t=v'$ be its nodes. By minimality of the path, $P_{v,v'} \cap (V(K) \cup N(K)) = \set{w_0,w_t}$.
    Let $u \in V(K) \setminus \set{u_v,u_{v'}}$. The set $U' = \set{u,u_v,u_{v'}} \cup V(P_{v,v'})$ induces a nice LDCC of size $t+4$. 

    Finally, $t \leq 2k-2$, meaning the LDCC has size at most $2k+2$. Indeed, consider the node $v'' = w_{\ceil{t/2}}$ in the middle of the path $P_{v,v'}$. This node cannot be at distance $< 1+\ceil{t/2}$ from $u_v$ in the graph $H$ by minimality of the path $P_{v,v'}$. Since $H$ has strong diameter $k$, $1+\ceil{t/2} \leq k$, hence the bound on the size of our nice LDCC in terms of $k$.

    \subparagraph*{Bigger cycles $C$ can also be extended to a nice LDCC}
    Consider now that the LDCC $H$ contains no triangle, but does contain an induced cycle. Let $C$ be such an induced cycle of minimum size $s \geq 4$. As in our argument for when $H$ contains a triangle, consider the neighbors of $C$, $N(C) \setminus V(C)$. This set is not empty since $H$ is connected and not a cycle.

    Suppose a node $v \in N(C) \setminus V(C)$ has multiple neighbors in $C$. By minimality of the cycle $C$, this is only possible with $C$ of size $4$ and $v$ having $2$ neighbors at opposite ends of $C$. In that case, $U' = V(C) \cup \set{v}$ induces a nice LDCC.
We now assume that each node $v \in N(C) \setminus V(C)$ has a single neighbor $u_v \in V(C)$. $\card{N(C) \setminus V(C)} \geq 2$ by the $2$-connectivity of $H$, since if it contained a single vertex $v$, removing $u_v$ would disconnect $H$.

    For each $v \in N(C) \setminus V(C)$, $v' \in N(C) \setminus (V(C)\cup \set{v})$, let $P_{v,v'}$ be a shortest path from $v$ to $v'$ in $H \setminus \set{u_v}$, guaranteed to exist by $H$ being $2$-connected, and $P'_{u_v,u_{v'}}$ be a shortest path from $u_v$ to $u_{v'}$ on $C$.
    Consider the cycles $C'_{v,v'} = G[V(P_{v,v'}) \cup P'_{u_v,u_{v'}}]$
induced by those two paths.
    Let $v,v'$ be two nodes generating a cycle $C'_{v,v'}$ of minimal size for all possible pairs $v,v'$. With $t$ the length of $P_{v,v'}$, the set $U' = V(C) \cup P_{v,v'}$ induces a nice LDCC of size $s + t+1$.

    We now bound the size of this nice LDCC in terms of $k$. 
    Let $s'\geq 1$ the length of $P'_{u_v,u_{v'}}$.
    By minimality of the cycle $C$ and definition of the cycle $C'_{v,v'}$, $s = \card{V(C)} \leq \card{V(C'_{v,v'})} = t+s'+2$.

    Let $w_0=v,w_1,\ldots,w_t = v'$ be the nodes of $P_{v,v'}$ and $w'_0=u_v,w'_1,\ldots,w'_{s'} = u_{v'}$ that of $P'_{u_v,u_{v'}}$. For each $(i,j) \in \set{0,\ldots,s'} \times \set{0,\ldots,t}$, $w'_i$ and $w_j$ are at distance $1+\min(i+j,s'+t-i-j)$ in $C'_{v,v'}$. This is the length of a shortest path between $w'_i$ and $w_j$ in $H$ by minimality of $C'_{v,v'}$. Set $i = \floor{s'/2}$ and $j = \floor{t/2}$. Since $H$ has diameter $k$ we get that $1+\min( \floor{s'/2} + \ceil{t/2} ,\ceil{s'/2} + \floor{t/2}) \leq k$. Thus $s' + t \leq 2k-1$. Therefore, the size of our nice LDCC satisfies $s+t+1 \leq 2(t+s')+3-s' \leq 4k$.

    \subparagraph*{An LDCC contains an induced cycle.}

    Consider $3$ connected nodes $u,u',u''$ in $H$. If they form a triangle, we are done.
    Otherwise, let them be arranged such that $\set{uu',u'u''}\subseteq E(G)$ and $uu'' \not \in E(G)$. Take a shortest path $P$ from $u$ to $u''$ in $H \setminus \set{u'}$, known to exist by the $2$-connectivity of $H$. Let $w_0 = u,w_1,\ldots,w_t=u''$ be the nodes of $P$. If there exists no edge between $u'$ and $w_i$ with $0<i<t$, $\set{u'}\cup V(P)$ induces a cycle. Otherwise, let $w_i$ be the node of the path $P$ of minimal index s.t.\ $u'w_i \in E(G)$. Then $\set{u'}\cup \set{w_j:0 \leq j \leq i}$ induces a cycle.
\end{proof}

\end{document}